\documentclass[11pt, onecolumn, draftcls]{IEEEtran}

\usepackage[latin1]{inputenc}
\usepackage[T1]{fontenc}
\usepackage[english]{babel}
\usepackage{graphicx}
\usepackage[caption=false,font=footnotesize]{subfig}
\usepackage{dsfont}
\usepackage{cite}
\usepackage{xfrac}
\usepackage{wrapfig}
\usepackage[cmex10]{amsmath}
\usepackage{amsfonts,amssymb}
\usepackage{amsthm}
\usepackage{mathrsfs}
\usepackage{booktabs}
\usepackage{color}
\usepackage{epstopdf}
\usepackage{float}
\usepackage{url}
\usepackage{multicol}
\usepackage{algorithm}              
\usepackage{algorithmic}
\usepackage{multirow}
\usepackage[usenames,dvipsnames]{pstricks}
\usepackage{epsfig}
\usepackage{pst-grad} 
\usepackage{pst-plot} 
\usepackage[space]{grffile} 
\usepackage{etoolbox} 
\makeatletter 
\patchcmd\Gread@eps{\@inputcheck#1 }{\@inputcheck"#1"\relax}{}{}
\makeatother

\newcommand{\bm}[1]{\mbox{\boldmath{$#1$}}}

\newtheorem{theorem}{Theorem}
\newtheorem{lemma}[theorem]{Lemma}

\begin{document}

\title{Improved Weighted Average Consensus in Distributed Cooperative Spectrum Sensing Networks}

\author{Aislan Gabriel Hernandes, Mario Proença Lemes Junior and Taufik Abrão
\thanks{Dept. of  Electrical Engineering (DEEL), State University of Londrina  (UEL),  Po.Box 10.011, Londrina, 86057-970, PR , Brazil. E-mails: taufik@uel.br \quad  aislangabrielhernandes@gmail.com}
\thanks{Dept. of  Computer Science (DC), State University of Londrina  (UEL),  Po.Box 10.011, Londrina, 86057-970, PR , Brazil, proenca@uel.br}
}
\maketitle

\begin{abstract}
This work proposes a fully distributed improved weighted average consensus (IWAC and WAC-AE) technique applied to cooperative spectrum sensing problem in cognitive radio systems. This method allows the secondary users cooperate based on only local information exchange without a fusion centre (FC). We have compared four rules of average consensus (AC) algorithms. The first rule is the simple AC without weights. The AC rule presents {performance comparable to the traditional cooperative spectrum sensing} (CSS) techniques, such as the equal gain combining (EGC) rule, which is a soft combining centralised method. Another technique is the weighted average consensus (WAC) rule using the weights based on the  SUs channel condition. This technique results in a performance similar to the maximum ratio combining (MRC) with soft combining (centralised CSS). Two new AC rules are analysed, namely weighted average consensus accuracy exchange (WAC-AE), and improved weighted average consensus (IWAC); the former relates the weights to the channel conditions of the SUs neighbours, while the latter combines the conditions of WAC and WAC-AE in the same rule. All methods are compared each other and with the hard combining centralised CSS. The WAC-AE results in a similar performance of WAC technique but with fast convergence, while  the IWAC can deliver suitable performance with small complexity increment{. Moreover,  IWAC method results in a similar convergence rate than the WAC-AE method but slightly higher than the AC and WAC methods}. Hence, the computational complexity of IWAC, WAC-AE, and WAC are proven to be very similar. The analyses are based on the numerical Monte-Carlo simulations (MCS), while algorithm's convergence is evaluated for both fixed and dynamic-mobile communication scenarios, and under AWGN and Rayleigh channels.
\end{abstract}


\section{Introduction} \label{Sec_Introduction}

Due to the growth of the wireless communication services, the available spectrum has become scarce. Measurements carried out by Federal Communications Commission (FCC) have demonstrated that the most of the allocated spectrum is not utilised \cite{fcc}. This motivates the use of the cognitive radio (CR) that has humanlike characteristics, such as, learning, adaptation and cooperation \cite{Mitola99}, \cite{Haykin05} which is able to increase the spectrum efficiency (SE) considerably. In a wireless regional area networks (WRANs), the main objective is to maximise the spectrum utilisation of the TV channels. The CR is the main technology in the WRAN IEEE Standard 802.22 \cite{Stevenson09}, which is applied in the white space TV channels.

One of the tasks realised by the CR is the spectrum sensing, that can be performed by means of single- or multi-band channel techniques; the latter being accomplished in multiple channels wideband scenarios. This task can be carried out in two ways, either in a non-cooperative manner, where secondary users sense independently the spectrum, or in a cooperative way, where the latter can be realised in a distributed or centralised way. In channel scenarios with shadowing and deep fading, the non-cooperative techniques result in poor performance. In such channel conditions, cooperative spectrum sensing (CSS) techniques are used, which allow the exchange of information between the elements of the network; hence, the channel severity  can be partially surpassed due to the diversity gain obtained with the CSS techniques, but with an increase in the complexity cost. In this sense, secondary users can be deployed as cooperative elements aiming {at establishing} decision-based on {\it hard combining} rules (AND, OR and Majority) or {\it soft combining}, including EGC and MRC {\it rules}. 

In the cooperative centralised mode, a fusion centre (FC) is deployed as the final decision maker for all secondary users. Moreover, relay nodes are widely applied in cooperative schemes employing the amplify-and-forward (AF) and decode-and-forward (DF) transmission protocols in a single-hop or multi-hop communication scheme. Usually, the multi-hop communication increases the energy efficiency compared to the single-hop schemes.

The performance of the centralised cooperative spectrum sense schemes operating under fading and AWGN channels is discussed in \cite{Ibnkahla14}. As well known, the hard combining presents degraded performance regarding soft combining rules. Among the hard combining rules, the more reliable performance is attained in {most cases} by the OR rule followed by Majority and AND rule; while among soft combining, the EGC always results in worst performance than MRC rule.

The term distributed (or decentralised) is defined as the way in which the decision is formed, implying in a local decision made by individual nodes. Thus, the term {\it distributed cooperative spectrum sensing} (DCSS) is defined as the final decision made from information exchanged between each node that previously made a local decision. There are some techniques in distributed/decentralised cooperative sensing, such as, belief propagation (BP) \cite{Wu14}, alternating direction method of multipliers (ADMM)\cite{Ding14}, and consensus algorithms (CA)\cite{Mitola11, Mitola15, Li10}. 

Recently, the consensus techniques have become promising in distributed cooperative sensing that allows the sensing without a proper FC receiver in a local one-hop neighbour communication.  The communication is based on bidirectional links (full duplex mode) and implies in a larger energy and spectrum efficiency and a smaller latency in the network. However, the  major part of existing techniques in the literature result in performance similar to the EGC centralised cooperative sensing, that is called simply average consensus (AC). In \cite{Mitola15}, it was proposed a novel consensus technique able to ensure a soft centralised cooperative sensing under the MRC rule. In \cite{Ashrafi11}, a binary consensus technique is developed to guarantee a superior performance to the quantised average consensus. Moreover, an average consensus (AC) technique applied to fixed and dynamic communication channels is discussed in \cite{Li09}. A {distributed average consensus} (DAC) is developed in \cite{Teguig15}, based on the goodness of fit test (GoF). This technique requires only the knowledge of the noise and using the Anderson Darling test \cite{AD52}. Furthermore, in \cite{Vosoughi16}, a trust-aware consensus is applied in the DCSS using Gossip algorithm. In \cite{Soatti16} a technique named  {\it weighted average consensus accuracy exchange}  (WAC-AE)  is proposed to solve the  localisation problem in networks equipped with several fixed nodes ensuring similar performance to the WAC and optimal ML, but with fast convergence. Moreover, in \cite{Nurellari16} a new consensus technique is applied in a quantised way, while in \cite{Kailkhura17} a new consensus technique is proposed to deal with security in a cognitive network in a system with byzantine attacks.

Against this background in the spectrum sensing methods, this paper proposes a two new AC techniques for cooperative descentralised spectrum sensing purpose, namely the {\it weighted average consensus accuracy exchange}  (WAC-AE)  and the {\it improved weighted average consensus} (IWAC). The IWAC method achieves the same performance of WAC method, which is similar to the optimal MRC combining, but with a competitive performance-complexity tradeoff. The WAC-AE is deployed in DCSS for the first time. The proposed IWAC method adopts similar conditions as that deployed in the WAC-AE and WAC rules. The advantage of IWAC lies on the lower number of iterations to achieve a target performance, which implies in a lower overall power consumption in the whole network. In summary, the contributions of this paper are threefold:

\begin{itemize} 
	\item The proposition of new rules on average consensus for distributed spectrum sensing purpose in the CRN context, namely IWAC and WAC-AE, which can achieve similar performance to the optimal centralised CSS with a small or similar number of iterations, depending on channel and system scenario;
	\item An analysis of convergence for the proposed consensus rules operating under fixed and dynamic network scenarios;
	\item A comparative complexity analysis of the proposed IWAC and WAC-AE regarding other AC rules;
\end{itemize}

The rest of the paper is organised as follows. The CR system model is presented in section \ref{sec:model}. The formulation of the centralised cooperative spectrum sensing and the fixed and dynamic channel communication model based on the graph theory are revisited in section \ref{CCSS}. In section \ref{CDSS} the existing average consensus techniques applied to distributed cooperative spectrum sensing are explored, while a novel distributed average consensus rule is formulated in section \ref{CDSS}. Numerical results supporting our finding are analysed in section \ref{sec:results}. Concluding remarks are offered in section \ref{sec:Concl}. For reference, and due to the large number of abbreviations deployed in this paper, a list of acronyms is summarized in Table \ref{tab:acron}.

 \begin{table}[!htpb]
	\renewcommand{\arraystretch}{1.0}
	\caption{Acronyms}
	\label{tab:acron}
	\centering
	\small
	\begin{tabular}{rl}\hline
		3C & Cooperative Consensus Convergence\\
		AC & Average Consensus\\
		ADMM & Alternating Direction Method of Multipliers\\
		AF & Amplify-and-Forward\\
		AWGN & Addictive White Gaussian Noise\\
		BF & Belief Propagation\\
		CLT & Central Limit Theorem\\
		CR & Cognitive Radio \\ 
		CSS & Cooperative Spectrum Sensing\\
		DAC & Distributed Average Consensus\\
		DCSS & Decentralised Cooperative Spectrum Sensing\\
		DF & Decode-and-Forward\\
		ED & Energy Detector\\
		EGC & Equal Gain Combining\\
		FC & Fusion Centre\\
		FCC & Federal Communication Commission\\
		GoF & Goodness-of-Fit\\
		IWAC & Improved Weighted Average Consensus \\
		MCS & Monte-Carlo Simulation\\
		MRC & Maximal Ratio Combining\\
		NLOS & Non-Line-of-Sight \\
		PU & Primary User\\
		ROC & Receiver Operating Characteristic\\	
		SE & Spectral Efficiency\\
		SLEM & Second Largest Eigenvalues Modulo\\
		SNR & Signal-Noise Ratio\\
		SU & Secondary User\\
		WAC & Weighted Average Consensus\\
		WAC-AE & Weighted Average Consensus Accuracy Exchange\\
		WRAN & Wireless Regional Area Network\\
\hline
\end{tabular}
\end{table}

\section{System Model}\label{sec:model}

We consider a cognitive wireless network with $N$ SUs and one PU (single-band system). All SUs sense the spectrum and cooperate with each other to determine the final decision. We can define two stages in the process: the sensing phase and the decision phase. In the sensing phase, each SU senses the spectrum. In this work, we adopt the energy detector (ED) because it requires lower design complexity and no {prior information} of the primary user (PU), but with a suboptimal performance. For the $i$-th SU, the received signal is defined as:
\begin{equation}
 y_{i} (t) =
  \begin{cases}
    n_{i} (t)       & \quad ,\mathcal{H}_{0}\\
    h_{i}s_{i}(t) + n_{i} (t)  &\quad ,\mathcal{H}_{1}\\
  \end{cases}
\end{equation}
where $\mathcal{H}_{0}$ is the hypothesis that the channel is idle, $\mathcal{H}_{1}$ is the hypothesis that the channel is busy, $y_{i} (t)$ is the received signal by the $i$-th SU, $s_{i}(t)$ is a BPSK modulated signal transmitted by the PU, $n_{i} (t)$ is the AWGN noise and $h_{i}$ is the amplitude channel gain that represents the multipath Rayleigh fading channel effect.

\subsection{Energy Detector} \label{ED}

Using the ED \cite{Urkowitz67}, each SU calculates a decision statistic $T_{i}$ over a detection interval of $N_s$ samples. The statistic test of the $i$-th SU can be written as:
\begin{equation}
	T_{i} = \sum_{t=0}^{N_s} {\vert y_{i} (t) \vert}^{2}.
\end{equation}
Hence, it is compared with a predefined threshold $\lambda$, and the decision of each user is:
\begin{equation}
 	T_{i} \underset{{\mathcal H_{0}}}{\overset{\mathcal H_{1}}{\gtrless}}  \lambda.
\end{equation}

The value $T_{i} \in \mathbb{R}^{+}$ under {AWGN channels presents a statistical distribution} given by \cite{Li10}:
\[ T_{i}  \sim 
  \begin{cases}
    \chi_{2TW}^2    & \quad ,\mathcal{H}_{0}\\
    \chi_{2TW}^2 (2\gamma)  &\quad ,\mathcal{H}_{1}\\
  \end{cases}
\]
where $  \chi_{2TW}^2$ and $ \chi_{2TW}^2 (2\gamma)$ is the central and non-central Chi-square distributions with $2TW = 2N_{s}$ degrees of freedom and non-centrality parameter of $2\gamma$.

Furthermore, under Rayleigh channels, the channel gain is random, and the distribution of the decision statistic becomes \cite{Li10}:
\[ T_{i}  \sim 
  \begin{cases}
    \chi_{2TW}^2    & \quad ,\mathcal{H}_{0}\\
    \chi_{2TW}^2 (2\gamma) + {\exp}(2 \overline{\gamma} + 2)  &\quad ,\mathcal{H}_{1}\\
  \end{cases}
\]
where the exponential distribution ${\exp}(2 \overline{\gamma} + 2)$ presents parameter $2 \overline{\gamma} + 2$. The $\overline{\gamma}$ is the average SNR and $\gamma$ is the instantaneous SNR.

Using the central limit theorem (CLT) for a large number of samples, the $i$-th statistic test $T_{i}$ is asymptotically normally distributed, with mean and variance given by \cite{Mitola15}:
\[ \mathbb{E}(T_{i}) =
  \begin{cases}
    N_{s} \sigma_{i}^2    & \quad ,\mathcal{H}_{0}\\
    (N_{s} + \eta_{i}) \sigma_{i}^2  &\quad ,\mathcal{H}_{1}\\
  \end{cases}
\]
\[ \text{var}(T_{i}) =
  \begin{cases}
    2 N_{s} \sigma_{i}^4    & \quad ,\mathcal{H}_{0}\\
    2(N_{s} + 2\eta_{i}) \sigma_{i}^4  &\quad ,\mathcal{H}_{1}\\
  \end{cases}
\]
where the $\sigma_{i}^2$ is the noise variance, while the $i$-th SNR of the SUs is given by:

\begin{equation}\label{eq:SNRi}
	{\eta_{i} = \sum_{t = 0}^{N_{s}} \frac{s_{i}^2 |h_{i}|^{2}}{\sigma_{i}^2} }.
\end{equation}

\section{Cooperative Spectrum Sensing} \label{CCSS}
Centralised versus distributed cooperative spectrum sensing strategies are revised in this section. Besides, dynamic communication channels are modelled with the aid of graph theory.

\subsection{Centralised Cooperative Spectrum Sensing }

Centralised cooperative {spectrum sensing methods need a fusion centre} (FC) to operate.
A cooperative network uses the SUs to sense the spectrum and an FC for the final decision.

In the FC, there are some ways to determine the final decision, including the hard combining, which can use different decision rule, such as the OR, Majority and AND  rules, and the soft combining way, that is based in EGC combining and MRC combining rules.

\subsubsection{Hard Decision} \label{HD}

In the hard combining spectrum sensing, $N$ cooperative SUs are sensing the total spectrum cooperatively; the final decision is given by the following metric, called final statistical test $T_{f}^{\textsc{hd}}$:
\begin{equation}
	T_{f}^{\textsc{hd}}  = \sum_{i=1}^{N} {d_i},
\end{equation}
where the ${d_i}$ is the decision of the $i$-th SU and ${d_i} \in \lbrace0,1\rbrace$, being ${d_i} = 0$ if PU is absent or ${d_i} = 1$ if the PU is present in the band. The performance is given in terms of probability of detection \cite{Ibnkahla14}:
\begin{equation}
	{\rm P}_d^{\textsc{hd}} = \sum_{q = i}^{N} {{N}\choose{q}} \left[ \prod_{\gamma = 1}^{q} {\rm P}_{d}^{\gamma} \cdot \prod_{\beta = 1}^{N-q} (1 - {\rm P}_{d}^{\beta}) \right].
\end{equation}

The {\it Or-And-Majority} rules allow to describe different ways to construct the threshold $\lambda$ in a hard combining centralised cooperative spectrum sensing scheme; in summary, 
\begin{itemize}
 \item{\it Or} rule: $\lambda = 1$. The rule {\it OR} ensure minimum interference to the PUs. The PU is considered present in a band, if only a single PU send $1$ to fusion centre in its decision,{\it i.e.,} if the statistic test of some SU add one. It can be seen that the OR rule is very conservative for the SUs to access the licensed band. As such, the chance of causing interference to the PU is minimised;
 \item{\it And} rule: $\lambda = N$, where $N$ means the number of collaborative nodes sensing the same sub-band. It is an aggressive rule, ensuring high rate of transmission to the SUs. The PU is considered present in the band, if and only if all CRs collaborative nodes sensing the presence of PU in the band;
 \item{\it Majority} rule: $\lambda = \left\lceil{\frac{N}{2}}\right\rceil$. The PU is considered present in the band, if the majority of SUs send $1$ to the FC. The function $\lceil{\cdot}\rceil$ is the ceil function.
 \end{itemize}
  
\subsubsection{Soft Decision} \label{SD}

The statistic test of the $i$-th SU is sent to the coordinator, the fusion center (FC), which collects all values of test statistic from all SUs. Then the overall statistic test $T_{f}^{\textsc{sd}}$ is calculated at the coordinator node as:
\begin{equation}
	T_{f}^{\textsc{sd}} = \sum_{i = 1}^{N} \rho_{i}T_{i}.
\end{equation}
If all $\rho_{i}$ is equal to each user, the cooperative technique has the equal gain combining (EGC) performance. If the values of $\rho_{i}$ is proportional to $\text{SNR}$, then the performance is same to maximum rate combining (MRC).

As in the case of cooperative SS, and following \cite{Mitola15}, the final decision $T_f$ is normally distributed, with mean and variance given by:
\begin{equation}
 \mathbb{E}(T_{f}^{\textsc{sd}}) =
  \begin{cases}
    \sum_{i = 1}^{N} \rho_{i} N_{s} \sigma_{i}^2  & \quad ,\mathcal{H}_{0}\\
    \sum_{i = 1}^{N} (N_{s}\sigma_{i}^2(1 + \eta_{i}))  &\quad ,\mathcal{H}_{1}\\
  \end{cases}
\end{equation}
\begin{equation}
 \text{var}(T_{f}^{\textsc{sd}}) =
  \begin{cases}
    \sum_{i = 1}^{N} \rho_{i}^2 2 N_{s} \sigma_{i}^4 & \quad ,\mathcal{H}_{0}\\
     \sum_{i = 1}^{N} \rho_{i}^2 (2N_{s}\sigma_{i}^4(1 + 2\eta_{i}))  &\quad ,\mathcal{H}_{1}\\
  \end{cases}
\end{equation}

As discussed in \cite{Nurellari16}, the performance of the centralised soft CSS can be evaluated for a given ${\rm P}_f$ as:
\begin{equation}\label{eq:Pcentralized}
	{\rm P}_{d}^{\textsc{c}} = Q\left(\frac{Q^{-1}({\rm P}_f)\sqrt{\text{var}(T_{f}^{\textsc{sd}}|\mathcal{H}_{0})} - \mathbb{E}(T_{f}^{\textsc{sd}}|\mathcal{H}_{1}) + \mathbb{E}(T_{f}^{\textsc{sd}}|\mathcal{H}_{0})}{\sqrt{\text{var}(T_{f}^{\textsc{sd}}|\mathcal{H}_{1})}}\right),
\end{equation}
where $Q(\cdot)$ is the Gaussian Q-function.

\subsection{Fixed and Dynamic DCSS Networks based on Graph Theory} \label{GT}
The fixed-nodes and mobile-node cooperative networks are modelled based on graph theory description. We define the elements of the network as the vertices and the communication links as the graph edges.

\subsubsection{Graph Theory Results}
To illustrate the graph theory-based description of a DCSS network, Fig. \ref{ex6} depicts an example of a distributed cooperative spectrum sensing (DCSS) network with 6 SUs keeping a bidirectional (full-duplex) one-hop communication. From the graph theory, this network presents $6$ vertices (or nodes) and $6$ edges.

\begin{figure}[!htbp]
	\centering
    \includegraphics[width=.42\textwidth]{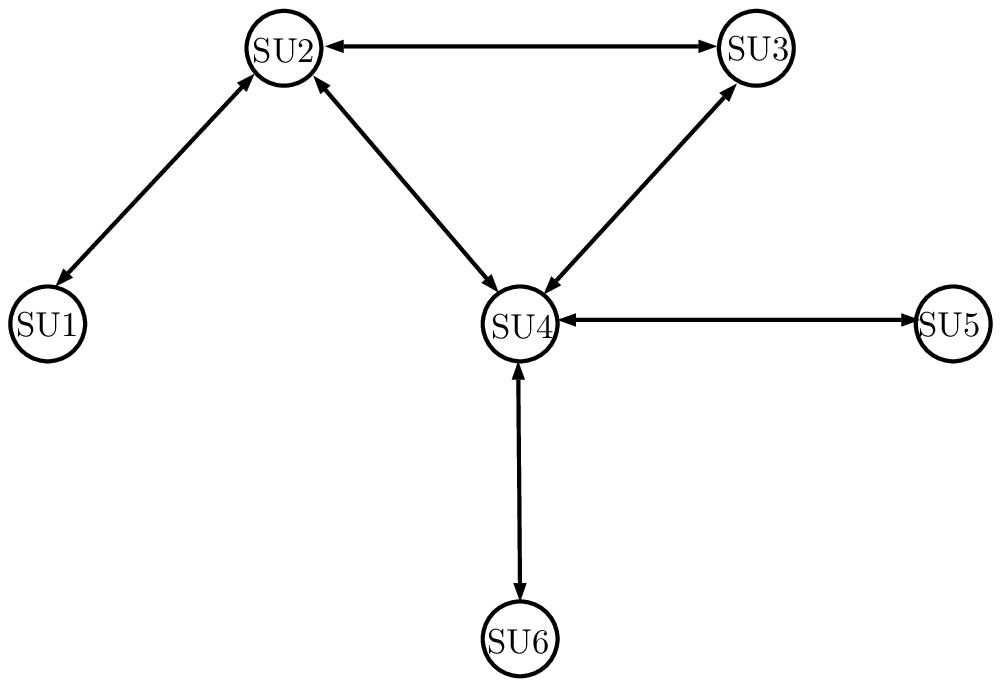}
	\caption{Descentralised cooperative scheme with $6$ SUs \cite{Kailkhura17}.}
	\label{ex6}
\end{figure}

In this paper, we will consider a decentralised network operating under fixed, as well as mobile communication channels.

\subsubsection{Fixed Communication Channel} \label{FG}

We consider that there are $N$ SUs interconnected and sharing the same channel bandwidth and links.  The network is modelled as a connected graph $\textsc{G} = (\mathscr{V}, \mathscr{E})$, where $\mathscr{V} = \lbrace1, 2, ..., N\rbrace$ is the vertices of the graph,{\it i.e.} the SUs contained in the network and $\mathscr{E} \subseteq \mathscr{V}\times \mathscr{V}$ is the edges, that representing the channel links between the SUs. The set of neighbors for the $i$-th SU is represented as $\mathscr{N}_{i} = \lbrace j \in \mathscr{V} : (i,j)\in \mathscr{E}\rbrace$, the cardinality (number of elements in the set) as $\aleph_{i}$ and the maximum cardinality as ${\rm max}(\aleph_{i})$.

The symmetric adjacent matrix of the graph $\mathscr{G}$ is ${\bf G} = [{\rm g}_{ij}]_{N \times N}$, where ${\rm g}_{ij} = 1$ if $(i,j) \in \mathscr{E}$, {\it i.e.}, when the $i$-th SU communicates with the $j$-th SU and ${\rm g}_{ij} = 0$ otherwise.

The Laplacian matrix of the graph $\mathscr{G}$ is defined as $\bf L = N - G$, where $\bf N$ is the maximum cardinality diagonal matrix of the graph defined as ${\bf N} = {\text{diag}}(\aleph_{1}, ..., \aleph_{N})$. Thus, the Laplacian matrix ${\bf L} = [l_{ij}]_{N \times N}$ can be constructed as:
\begin{equation} \label{lij}
 l_{ij} =
  \begin{cases}
    \aleph_{i}      & \quad ,{\text {if}} \quad i=j\\
    -1  &\quad ,{\text {if}} \quad j\in{\mathscr{N}_{i}}\\
    0 &\quad ,{\text {otherwise.}}
  \end{cases}
\end{equation}

To illustrate those definitions, the network presented in Fig. \ref{ex6}), which will be analysed in section \ref{NT6}, defines the following diagonal matrix with maximum cardinality:
\begin{equation} 
{\bf N}_6 = 
\begin{bmatrix}
    1 & 0 & 0 & 0 & 0 & 0  \\
    0 & 3 & 0 & 0 & 0 & 0  \\
    0 & 0 & 2 & 0 & 0 & 0  \\
    0 & 0 & 0 & 4 & 0 & 0  \\
    0 & 0 & 0 & 0 & 1 & 0  \\
    0 & 0 & 0 & 0 & 0 & 1  
\end{bmatrix},
\label{matrix6a}
\end{equation}
and the adjacency matrix takes the form:
\begin{equation} \label{matrix6b}
{\bf G}_6 = 
\begin{bmatrix}
    0 & 1 & 0 & 0 & 0 & 0  \\
    1 & 0 & 1 & 1 & 0 & 0  \\
    0 & 1 & 0 & 1 & 0 & 0  \\
    0 & 1 & 1 & 0 & 1 & 1  \\
    0 & 0 & 0 & 1 & 0 & 0  \\
    0 & 0 & 0 & 1 & 0 & 0  
\end{bmatrix}.
\end{equation}
Therefore, the Laplacian matrix for this network is given by:
\begin{equation} 
{\bf L}_6 = 
\begin{bmatrix}
    1 & -1 & 0 & 0 & 0 & 0  \\
    -1 & 3 & -1 & -1 & 0 & 0  \\
    0 & -1 & 2 & -1 & 0 & 0  \\
    0 & -1 & -1 & 4 & -1 & -1  \\
    0 & 0 & 0 & -1 & 1 & 0  \\
    0 & 0 & 0 & -1 & 0 & 1  
\end{bmatrix}.
\label{matrix6c}
\end{equation}

\subsubsection{Dynamic Communication Channel} \label{DG}
Similarly to the static communication channel, in the dynamic channel case,  the Laplacian matrix of the graph $\mathscr{G} (k)$ is defined as ${\bf L}(k) = {\bf N - G}(k)$, where $k$ is an integer that represents the time of network change,{\it i.e.,} the graph positions changes according to the time integer intervals, $\bf N$ is the maximum cardinality diagonal matrix of the graph defined as ${\bf N} = {\text{diag}}(\aleph_{1}, ..., \aleph_{N})$. Thus, the Laplacian matrix ${\bf L}(k) = [l_{ij}]_{N \times N}$ can be constructed similarly as (\ref{lij}).

A better description of the dynamic channel can be made taking into account a probability of connection (in the 
neighbours communication sense) that can be described by the {\it a priori} probability ${\rm Pr}_{\text{connection}} \in [0,1]$. The probability of link failure is ${\rm Pr}_{\text{fail}} = 1 - {\rm Pr}_{\text{connection}}$. When this probability is zero the channel is fixed and otherwise the network presents some mobility. Hence, the structure of the Laplacian matrix is ready modified considering the {\it a priori} probability of connection as:
\begin{equation}
 l_{pij} =
  \begin{cases}
    \sum_{j = 1}^{N} {\rm Pr}_{\text{connection}}     & \quad ,{\text {if}} \quad i=j\\
    -{\rm Pr}_{\text{connection}}  &\quad ,{\text {if}} \quad j\in{\mathscr{N}_{i}}\\
    0 &\quad ,{\text {otherwise.}}
  \end{cases}
\end{equation}

\section{Consensus-based Distributed Cooperative Spectrum Sensing} \label{CDSS}
Existing distributed consensus-based fusion techniques only ensure EGC performance; {such techniques are identified} as average consensus algorithm (AC) \cite{Mitola15}. Therefore, the EGC performance is inferior regarding the centralised MRC combining (optimal combining) schemes. Based on this, new consensus algorithms have been proposed in the literature to ensure MRC performance. {These} algorithms are {denominated} weighted average consensus (WAC) techniques \cite{Mitola15}. The performance of the WAC technique is closed to the MRC centralised combining (soft combining). However, the WAC algorithm has slow convergence when the case of unbalanced SNR at different SUs, that are directly related to the weights design.

\subsection{Average Consensus} \label{ACS}
In the average consensus (AC) method the estimation of the $i$-th SU energy is updated at the iteration time $k = 1, 2, ...$ according to the rule \cite{Yu10}:
\begin{equation} 
x_{i} (k + 1) = x_{i}(k) + \alpha \sum_{j \in \mathscr{N}_{i}}^{} {\rm g}_{ij}(x_{j} (k) - x_{i} (k)),
\end{equation}
where $\alpha$ is the iteration step size satisfying $0 < \alpha < ({\max(\aleph_{i})})^{-1}$. The elements of the adjacent matrix $ {\rm g}_{ij}$ define de network topology.

The initial statistic before the fusion at the iteration $k = 0$ is considered as $x_{i} (0) = T_{i}$.

For the AC method, the final convergence is obtained as \cite{Mitola15}:
\begin{equation}
	x_{i}(k) \rightarrow x^{*} = \frac{\sum_{i=1}^{N} x_{i}(0)}{N}, \qquad \text{when} \quad k \rightarrow \infty,
\end{equation} 
while the final decision is compared with a pre-defined threshold $\lambda$ and has the form:
\begin{equation} 
 {\text {Decision}} =
  \begin{cases}
      \mathcal{H}_{0}   & \quad , x^{*} > \lambda\\
    \mathcal{H}_{1}  &\quad , {\text {otherwise.}}
  \end{cases}
\end{equation}

In the compact vector-matrix form, the rule can be described as:
\begin{equation}
	{\bf x}(k+1) = {\bf P}_\textsc{ac}{\bf x}(k),
\end{equation}
where ${\bf P}_\textsc{ac} = {\bf I} - \alpha({\bf N} - {\bf G})$ is the Perron matrix and can be written also as ${\bf P}_\textsc{ac} = {\bf I} - \alpha{\bf L}_\textsc{ac}$. Here, the Laplacian matrix is ${\bf L}_\textsc{ac} = {\bf L}$, as defined in the last section. Hence, the performance regarding probability of detection, for a given fail probability at the  $i$-th SU, can be described in the same way of Eq. \eqref{eq:Pcentralized},
but now considering distributed soft CSS decisions.

Algorithm \ref{algo:AC} describes a pseudocode of AC method.

\begin{algorithm}[!htbp]
	\caption{ - Average Consensus - (AC)} \label{algo:AC}
	\begin{algorithmic}[1]
	\small
		\STATE \textbf{Input: $\alpha$, $K$, $\bf{T}$}
		\FOR {$k=0$ \textbf{to} $K - 1$}
		\STATE {${\bf{x}}(0) = \bf{T}$}
		\STATE {${\bf P}_\textsc{ac} = {\bf I} - \alpha {\bf L}_\textsc{ac}$}	
		\STATE {${\bf x}(k+1) = {\bf P}_\textsc{ac}{\bf x}(k)$}
		\ENDFOR
		\STATE \textbf{Output: $\bf{x}$}
	\end{algorithmic}
\end{algorithm}

\subsection{Weighted Average Consensus} \label{WACS}
The weighted average consensus (WAC) rule can approach to soft combining performance (MRC). The WAC rule is given by \cite{Mitola15}, \cite{Mitola11}:
\begin{equation}
	x_{i} (k + 1) = x_{i}(k) + \frac{\alpha}{\omega_{i}} \sum_{j \in \mathscr{N}_{i}}^{} {\rm g}_{ij}(x_{j} (k) - x_{i} (k)),
\end{equation}
where $\omega_{i}$ is the weighted ratio according to the channel condition of the $i$-th SU and $\alpha$ is the iteration step size satisfying $0 < \alpha < \left({\rm{max}}(\aleph_{i})\right)^{-1}$. The final convergence is obtained as \cite{Mitola15}:
\begin{equation}
	x_{i}(k) \rightarrow x^{*} = \frac{\sum_{i=1}^{N}\omega_{i}x_{i}(0)}{\sum_{i=1}^{N}\omega_{i}}, \qquad \text{when} \quad k \rightarrow \infty.
\end{equation} 
Moreover, when the values of $\omega_{i}$ is equal to the all SUs, the final convergence is similar to EGC combining,{\it i.e.,} the same of the AC method.

In the WAC algorithm, the weights are related to the channel conditions of the $i$-th SU. According \cite{Mitola15}, a  {\it sub-optimal weights} for the WAC spectrum sensing receiver operating under Rayleigh fading channels can obtain as an estimative of the SNR state channel:
\begin{equation}\label{eq:wi_WAC}
	\omega_{i} = \frac{1}{2 \ell} \sum_{\wp = k - \ell}^{k} (T_{i,\wp} - 2N_{s}),
\end{equation}
where $\ell$ is the length of the estimation window and $T_{i,\wp}$ is the $\wp$-th measurement (statistic test) of the $i$-th SU.

For the AWGN channel, the optimal weights are simply calculated solving an optimisation problem that maximises the {\it deflection coefficient} {\cite{Mitola15}}:
\begin{equation}\label{eq:deflection_coefs}
	\omega_{i} = \frac{\eta_i}{\sigma_{i}^{2}},
\end{equation}
where $\eta_i$ is defined in \eqref{eq:SNRi}. 

Using the WAC in the compact form, the discrete consensus rule can be represented in the vector-matrix form as \cite{Mitola15}:
\begin{equation}
	{\bf x}(k+1) = {\bf P}_\textsc{wac}{\bf x}(k),
\end{equation}
where the Perron matrix can be written as ${\bf P}_\textsc{wac} = {\bf I} - \alpha{\bf \Delta}^{-1}{\bf L}_\textsc{wac}$. The diagonal matrix ${\bf \Delta} = \text{diag}(\omega_{1}, ..., \omega_{N})$ is the weight diagonal matrix. Here, the Laplacian matrix ${\bf L}_\textsc{wac} = {\bf L}$.

The performance can be obtained in the same way of Eq. \eqref{eq:Pcentralized}, but now considering distributed soft decisions. The pseudocode for the WAC algorithm is depicted in Algorithm \ref{algo:WAC}:

\begin{algorithm}[!htbp]
	\caption{ - Weighted Average Consensus - (WAC)} \label{algo:WAC}
	\begin{algorithmic}[1]
	\small
		\STATE \textbf{Input: $\alpha$, $K$, $\bf{\Delta}$, $\bf{T}$}
		\FOR {$k=0$ \textbf{to} $K - 1$}
		\STATE {${\bf{x}}(0) = \bf{T}$}
		\STATE {${\bf P}_\textsc{wac} = {\bf I} - \alpha\bf{\Delta}^{-1}{\bf L}_\textsc{wac}$}	
		\STATE {${\bf x}(k+1) = {\bf P}_\textsc{wac}{\bf x}(k)$}
		\ENDFOR
		\STATE \textbf{Output: $\bf{x}$}
	\end{algorithmic}
\end{algorithm}

\subsection{Weighted Average Consensus Accuracy Exchange}

Recently, the {\it weighted average consensus accuracy exchange} (WAC-AE) was proposed \cite{Soatti16} and  \cite{Kailkhura17} in a different context treated herein, i.e., respectively to solve the  localisation problem in networks equipped with several fixed nodes and to deal with security issues in a cognitive network. In the new context of DCSS,  the WAC-AE rule to is given by:
\begin{equation}
x_{i} (k + 1) = x_{i}(k) + \alpha \sum_{j \in \mathscr{N}_{i}}^{} \omega_{j} {\rm g}_{ij} (x_{j} (k) - x_{i} (k)),
\end{equation}
where $\omega_{j}$ is the weighted ratio according to the channel condition of the $j$-th SUs. The convergence is guaranteed taking the step size among $0 < \alpha < \left({\rm{max}}_{i} \sum_{j \in \mathscr{N}_{i} } \omega_{j}\right)^{-1}$. The associated final convergence is obtained as:
\begin{equation}
	x_{i}(k) \rightarrow x^{*} = \frac{\sum_{i=1}^{N}\omega_{i}x_{i}(0)}{\sum_{i=1}^{N}\omega_{i}}, \qquad \text{when} \quad k \rightarrow \infty.
\end{equation} 

In the WAC-AE algorithm the weights are related to the channel conditions of the $j$-th SUs neighbours. Adopting the {\it sub-optimal weights} for Rayleigh channels, results:
\begin{equation}
	\omega_{j} = \frac{1}{2 \ell} \sum_{\wp = k - \ell}^{k} (T_{j,\wp} - 2N_{s}),
\end{equation}
where $\ell$ is the length of the estimation window and $T_{j,\wp}$ is the $\wp$-th measurement (statistic test) of the $j$-th SUs. Besides, for the AWGN channel, the optimal weights are simply calculated as in \eqref{eq:deflection_coefs}.

In the compact form, the discrete WAC-AE consensus rule can be represented in the vector-matrix form as:
\begin{equation}
	{\bf x}(k+1) = {\bf P}_\textsc{wac-ae}{\bf x}(k),
\end{equation}
where the Perron matrix is ${\bf P}_\textsc{wac-ae} = {\bf I} - \alpha{\bf L}_\textsc{wac-ae}$. The modified Laplacian matrix ${\bf L}_\textsc{wac-ae} = [l_{ij_{\textsc{wac-ae}}}]_{N \times N}$ is construct as:

\begin{equation} 
 l_{ij_{\textsc{wac-ae}}} =
  \begin{cases}
    \sum_{j \in \mathscr{N}_{i}}^{} \omega_{j}     & \quad ,{\text {if}} \quad i=j\\
    -\omega_{j}  &\quad ,{\text {if}} \quad j\in{\mathscr{N}_{i}}\\
    0 &\quad ,{\text {otherwise.}}
  \end{cases}
\end{equation}

The pseudocode of the WAC-AE is presented in Algorithm \ref{algo:WAC-AE}.

\begin{algorithm}[!htbp]
	\caption{ - Weighted Average Consensus - Accuracy Exchange - (WAC-AE)} \label{algo:WAC-AE}
	\begin{algorithmic}[1]
	\small
		\STATE \textbf{Input: $\alpha$, $K$, $\bf{\Delta}$, $\bf{T}$}
		\FOR {$k=0$ \textbf{to} $K - 1$}
		\STATE {${\bf{x}}(0) = \bf{T}$}
		\STATE {${\bf P}_\textsc{wac-ae} = {\bf I} - \alpha\bf{\Delta}^{-1}{\bf L}_\textsc{wac-ae}$}	
		\STATE {${\bf x}(k+1) = {\bf P}_\textsc{wac-ae}{\bf x}(k)$}
		\ENDFOR
		\STATE \textbf{Output: $\bf{x}$}
	\end{algorithmic}
\end{algorithm}

\section{Improved Weighted Average Consensus} \label{IWACS}
In this section, we propose a new rule to weighted average consensus for distributed cooperative spectrum sensing purpose. The new rule improves the weighted average consensus (IWAC), being described by the following updating equation:
\begin{equation}
x_{i} (k + 1) = x_{i}(k) + \frac{\alpha}{\omega_{i}} \sum_{j \in \mathscr{N}_{i}}^{} \omega_{j} {\rm g}_{ij} \left[x_{j} (k) - x_{i} (k)\right],
\end{equation}
where $\omega_{j}$ is the weighted ratio according to the channel condition of the $j$-th SUs and $\omega_{i}$ is the weight according to the channel condition of the $i$-th SU. The convergence is guaranteed taking the step size in the interval:
\begin{equation}
0 < \alpha < \left({\rm{max}}_{i} \sum_{j \in \mathscr{N}_{i} }\omega_{j}\right)^{-1}
\end{equation}
The final convergence to the IWAC method is obtained as:
\begin{equation}
	x_{i}(k) \rightarrow x^{*} = \frac{\sum_{i=1}^{N}\omega_{i}x_{i}(0)}{\sum_{i=1}^{N}\omega_{i}}, \qquad \text{when} \quad k \rightarrow \infty.
\end{equation} 

Moreover, we can adopt the same {\it sub-optimal weights} of the WAC rule, \eqref{eq:wi_WAC}, for the distributed cooperative SSNs operating under Rayleigh fading channels as:
\begin{equation}
	\omega_{\xi} = \frac{1}{2 \ell} \sum_{\wp = k - \ell}^{k} (T_{\xi,\wp} - 2N_{s}),
\end{equation}
where $\ell$ is the length of the estimation window, $\xi \in (i,j)$, $T_{\xi,\wp}$ is the $\wp$-th measurement (statistic test) of SU. Again, for the AWGN channel the weights are calculated as in \eqref{eq:deflection_coefs}, {\it i.e.}, $\omega_{\xi} = \frac{\eta_\xi}{\sigma_\xi^{2}}$. 

In the compact form, the discrete consensus rule can be represented in the vector-matrix form as:
\begin{equation}\label{eq:updating_IWAC}
	{\bf x}(k+1) = {\bf P}_\textsc{iwac}{\bf x}(k),
\end{equation}
where the modified Perron matrix now is defined as: 
\begin{equation}\label{eq:Perron_IWAC}
	{\bf P}_\textsc{iwac} = {\bf I} - \alpha{\bf \Delta}^{-1}{\bf L}_{\textsc{iwac}}.
\end{equation}

In the proposed IWAC spectrum sensing, the modified Laplacian matrix ${\bf L}_\textsc{iwac} = [l_{ij_{\textsc{iwac}}}]_{N \times N}$ is construct as
\begin{equation} 
 l_{ij_{\textsc{iwac}}} =
  \begin{cases}
    \sum_{j \in \mathscr{N}_{i}}^{} \omega_{j}     & \quad ,{\text {if}} \quad i=j\\
    -\omega_{j}  &\quad ,{\text {if}} \quad j\in{\mathscr{N}_{i}}\\
    0 &\quad ,{\text {otherwise.}}
  \end{cases}
\end{equation}

The matrix ${\bf \Delta} = \text{diag}(\omega_{1}, ..., \omega_{N})$ is the weight diagonal matrix.
Notice that the {\it receiver operator characteristics} (ROC) performance for the IWAC spectrum sensor can be obtained in a same way of Eq. \eqref{eq:Pcentralized}, but taking into account distributed soft CSS decisions, as discussed in subsection \ref{sec:ROC}. 

A pseudo-code for the IWAC implementation considering static and dynamic channel environments is presented in Algorithm \ref{algo:IWAC}.

\begin{algorithm}[!htbp]
	\caption{ - Improved Weighted Average Consensus - (IWAC)} \label{algo:IWAC}
	\begin{algorithmic}[1]
	\small
		\STATE \textbf{Input: $\alpha$, $K$, $\bf{\Delta}$, $\bf{T}$}
		\FOR {$k=0$ \textbf{to} $K - 1$}
		\STATE {${\bf{x}}(0) = \bf{T}$}
		\STATE {${\bf P}_\textsc{iwac} = {\bf I} - \alpha\bf{\Delta}^{-1}{\bf L}_\textsc{iwac}$}	
		\STATE {${\bf x}(k+1) = {\bf P}_\textsc{iwac}{\bf x}(k)$}
		\ENDFOR
		\STATE \textbf{Output: $\bf{x}$}
	\end{algorithmic}
\end{algorithm}

\subsection{Convergence Analysis for the IWAC Algorithm}\label{sec:converg_IWAC}
In this section, the convergence analysis for the IWAC algorithm is developed taking into account both system scenarios, static and dynamic SU's in the CR networks.

\subsubsection{Fixed Networks}
Using the IWAC in the compact form, the discrete consensus rule can be represented in the vector-matrix form by the updating equation \eqref{eq:updating_IWAC}, where the Perron matrix ${\bf P}_{\textsc{iwac}}$ is given by \eqref{eq:Perron_IWAC}.  

The IWAC rule convergence depends on the convergence of the infinite stochastic matrix product. Based on the Perron-Frobenius Theorem \cite{Mitola15}, \cite{Horn85} we find:
\begin{equation}\label{eq:Prod_Perron_IWAC}
	{\bf P}_{\infty_{\textsc{iwac}}}  =  \mathop{\lim}\limits_{k \to\infty} \prod_{\ell = 1}^{k} {\bf P}_{\ell_{\textsc{iwac}}} = \frac{{\bf 1}{\bm \omega}^{T}}{\bm{\omega}^{T}{\bf 1}},
\end{equation}
where $\bm{\omega}^T = [\omega_{1}\, \omega_{2}\, \ldots\, \omega_{N}]$ and vector $\bm{1} = [1\, 1\, \ldots\, 1]^T$ has dimension $N\times 1$. 

The proof can be obtained considering that the matrix ${\bf P}_{\textsc{iwac}}$ is a primitive non-negative matrix, {\it i.e.,} the $k$-th power is positive for some natural number $k$ with left and right eigenvectors ${\bf u}$ and ${\bf v}$, respectively, that satisfy ${\bf P_{\textsc{iwac}}}{\bf v} = {\bf v}$ and ${\bf u}^{T}{\bf P_{\textsc{iwac}}} = {\bf u}^{T}$. The Perron-Frobenius Theorem ensures that $\lim_{k\rightarrow\infty} \prod_{\ell = 1}^{k} {\bf P}_{\ell_{\textsc{iwac}}}  = \frac{{\bf vu}^{T}}{{\bf v}^{T}{\bf u}}$.

\begin{lemma}
Let $\mathscr {G}$ a connected graph with $N$ vertices. The Perron matrix $\bf P_{\textsc{iwac}}$, with $0 < \alpha < {({{\rm{max}}_{i} \sum_{j \in \mathscr{N}_{i} }\omega_{j}})^{-1}}$ has the following properties:
\begin{itemize}
\item[\textsc{p.1.}] The Perron matrix $\bf P_{\textsc{iwac}}$ is a nonnegative matrix with left eigenvector $\bm{\omega}$ and right eigenvector ${\bf 1}$;
\item[\textsc{p.2.}] All eigenvalues of Perron matrix $\bf P_{\textsc{iwac}}$ are in a unit circle;
\item[\textsc{p.3.}] The Perron matrix $\bf P_{\textsc{iwac}}$ is a primitive matrix.
\end{itemize}
\end{lemma}

\begin{proof}
The first property is based on that ${\bf P_{\textsc{iwac}}} {\bf 1} = {\bf 1} - \alpha{\bm \Delta}^{-1}{\bf L}_{{\textsc{iwac}}} {\bf 1} = {\bf 1}$ and ${\bm\omega}^{T} {\bf P}_{\textsc{iwac}}  = {\bm \omega}^{T} - \alpha {\bm \omega}^{T} {\bm \Delta}^{-1}{\bf L}_{{\textsc{iwac}}}  = {\bm \omega}^{T}$ that implies in a left eigenvector $\bm{\omega}$ and a right eigenvector ${\bf 1}$. 
	
\noindent The second property is guaranteed by the Gershgorin Theorem and the third property is guaranteed by the step size $\alpha$ of the IWAC method.
\end{proof}

\begin{theorem}\label{theo:1}
For the IWAC iterative process, the step size $\alpha$ satisfies the condition $0 < \alpha < {({\rm{max}}_{i} \sum_{j \in \mathscr{N}_{i} }^{} \omega_{j})^{-1}}$, in which the elements $\omega_{i}$ and $\omega_{j}$ operating in a fixed communication network occur infinitely (infinite iterations, fixed values); hence, the iteration converges to
	\begin{equation}
		\lim_{k\rightarrow\infty} x_{i} (k) = \frac{\sum_{i = 1}^{N} \omega_{i} x_{i} (0)}{\sum_{i = 1}^{N} \omega_{i}}.
	\end{equation}
\end{theorem}
\begin{proof}
The IWAC consensus method achieves asymptotically the convergence and the Perron-Frobenius Theorem ensures that the limit $\mathop{\lim}\limits_{k\rightarrow\infty} \prod_{\ell = 1}^{k} {\bf P}_{\ell_{\textsc{iwac}}}$ exists for primitive matrices, then
\begin{equation}
 \begin{split}
 {\bf x}(k+1) & = {\bf P}_\textsc{iwac}{\bf x}(k),\\
{\bf x}^{\ast}  = \lim_{k\rightarrow\infty} {\bf x}(k+1)  & = \mathop{\lim}\limits_{k\rightarrow\infty} {\prod}_{\ell = 1}^{k} {\bf P}_{\ell_{\textsc{iwac}}} {\bf x}(0),\\  
	                 {\bf x}^{\ast}	  & = \frac{{\bf 1}{\bm \omega}^{T}}{\bm{\omega}^{T}{\bf 1}} {\bf x}(0),\\ 
\text{where} \quad x_i^*	                 	  &  = \frac{\sum_{i = 1}^{N} \omega_{i} x_{i} (0)}{\sum_{i = 1}^{N} \omega_{i}}.
  \end{split}
\end{equation}
\end{proof}

\subsubsection{Dynamic Networks}
For a network with $N$ SUs, there are a finite number of possible graphs (for example, $r$ graphs). We denote the set of possible graphs $\lbrace \mathscr{G}_{1}, ..., \mathscr{G}_{r}\rbrace$ and there are a correspondent set of Perron matrices $\lbrace {\bf P}_{\textsc{iwac}}^{1}, ..., {\bf P}_{\textsc{iwac}}^{r}\rbrace$. Considering that $1\leq s \leq r$. The {\it weighted average consensus rule} is given by:
\begin{equation}
	{\bf x}(k+1) = {\bf P}_{\textsc{iwac}}^{s(k)}{\bf x}(k).
\end{equation}

The proof for dynamic network follows the fact that the IWAC consensus iteration is a paracontraction\footnote{A paracontraction is a process at where $||{\bf P}_{\textsc{iwac}}{\bf x}|| \leq ||{\bf x}|| \Leftrightarrow {\bf P}_{\textsc{iwac}}{\bf x} \neq {\bf x}$ is guaranteed.} process with fixed points building by the eigenspaces of the Perron matrices. 

For the connected graph $\mathscr{G}(k)$ and the Perron matrix $ \bf P_{\textsc{iwac}}$, being that a nonnegative primitive matrix, having $\bm{\omega}$ and ${\bf 1}$ as the left and right eigenvector respectively. For a paracontracting matrix, we denote the subspace $\mathbb{H({\bf P_{\textsc{iwac}}})}$, that is an eigenspace associated with eigenvalue $1$. The collection of graphs  $\lbrace \mathscr{G}_{1}, ..., \mathscr{G}_{r}\rbrace$ are connected and occur infinitely, the Perron matrices satisfy $\bigcap_{z = 1}^{r} \mathbb{H} ({{{\bf P}^{z}}_{\textsc{iwac}}}) = \text{span{({\bf 1})}}$. From the properties of the paracontracting process, the subspace is fixed, then the iterative process has a limit, that is guaranteed by the Perron-Frobenius Theorem that ensures the  asymptotic convergence.

Hence the following Theorem guarantees the convergence of the IWAC procedure operating under dynamic distributed cooperative spectrum sensing  networks.

\begin{theorem}\label{theo:3} 
For the IWAC iterative process, the step size $\alpha$ satisfying $0 < \alpha < {({\rm{max}}_{i} \sum_{j \in \mathscr{N}_{i} }^{} \omega_{j})^{-1}}$, with weight elements $\omega_{i}$ and $\omega_{j}$ for a dynamic cooperative communication occurring infinitely (infinite iterations), the IWAC rule converges to:
\begin{equation}
  \begin{split}
  x_i^* &= \lim_{k\rightarrow\infty} x_{i} (k)  = \frac{\sum_{i = 1}^{N} \omega_{i} x_{i} (0)}{\sum_{i = 1}^{N} \omega_{i}}\\
\text{or} \quad  {\bf x}^{\ast}	  & = \frac{{\bf 1}{\bm \omega}^{T}}{\bm{\omega}^{T}{\bf 1}} {\bf x}(0).
  \end{split}
\end{equation}
\end{theorem}

\begin{proof}
The proof is similar to the fixed network case, given that the Perron-Frobenius applies. Hence, the proof is omitted.\\
\end{proof}

Should be observed that the convergence of the fixed and dynamic communications, results in the same final result. Numerical evidence corroborating this fact is presented in section \ref{Sec_Results}.

\section{Numerical Results}\label{sec:results} \label{Sec_Results}
In this section, we have compared the performance of various spectrum sensors discussed in this work. We have considered four scenarios, all of them with one primary user, PU$=1$. In Scenario A, the network is fixed, {\it i.e.,} the SUs are considered static in the same position during all DCSS process. The channel is considered only under AWGN noise effect, where the SUs SNRs are contained in a range of $[-10, 0]$ dB. The Monte-Carlo simulations (MCS) have been realised considering a network with $6$ and $10$ SUs. In the Scenario B, we consider $10$ and $20$ SUs in the network in a AWGN channel with SNRs between $[-10, 0]$ dB. Now the scenario is dynamic, {\it i.e.,} the SUs has mobility in the network. In the Scenario C, the channel is Rayleigh with SNR $\in [-2, 5]$ dB. Furthermore, the SUs are fixed and the simulations consider $6$ and $10$ SUs. Finally, in the Scenario D, the network is dynamic under Rayleigh channel and SNRs values between $[-2, 5]$ dB; $10$ and $20$ SUs have been considered in the simulations. In Rayleigh channels, we have considered the weights $\omega_i$ as a perfect estimation of the average SNRs in each node. The main system parameters for the Scenarios A to D are summarised in Table \ref{tab:scen}.
 
\begin{table}[!htbp]
\centering
\caption{System Scenarios, considering PU $=1$ user}
\begin{tabular}{ll}
	\hline
	\bf Parameter & \bf Adopted Values \\
	\hline
	\multicolumn{2}{c}{Scenario \bf A}\\
	\hline
	Channel & AWGN\\
	Network Type & Fixed, ${\rm Pr}_{\text{fail}} = 0$ \\	
	Secondary users & SU $\in \{6, 10\}$ users\\
	Range of SNR & SNR$_{\textsc{su}}  \in  \lbrace 0, \,\, -10\rbrace$ [dB]\\
	\hline
	\multicolumn{2}{c}{Scenario \bf B}\\
	\hline
	Channel & AWGN\\
	Network Type & Dynamic, ${\rm Pr}_{\text{fail}} = 0.4$ \\	
	Secondary users & SU $\in \{10, 20\}$ users\\
	Range of SNR & SNR$_{\textsc{su}}  \in  \lbrace 0, \,\, -10\rbrace$ [dB]\\
	\hline
	\multicolumn{2}{c}{Scenario \bf C}\\
	\hline
	Channel & flat Rayleigh\\
	Network Type & Fixed, ${\rm Pr}_{\text{fail}} = 0$ \\	
	Secondary users & SU $\in \{6, 10\}$ users\\
	Range of SNR & SNR$_{\textsc{su}}  \in  \lbrace -2, \,\, 5\rbrace$ [dB]\\
	\hline
	\multicolumn{2}{c}{Scenario \bf D}\\
	\hline
	Channel & flat Rayleigh\\
	Network Type & Dynamic, ${\rm Pr}_{\text{fail}} = 0.4$\\ 
	Secondary users & SU $\in \{10, 20\}$ users\\
	Range of SNR & SNR$_{\textsc{su}}  \in  \lbrace -2, \,\, 5\rbrace$ [dB]\\
	\hline
\end{tabular}

\label{tab:scen}
\end{table}
 Table \ref{tab_ref_values1} depicts the main adopted simulation parameters values. These values are adopted by all scenarios. For each MCS $5000$ realisations have been considered, with $12$ samples per decision and a fail probability communication between SUs in the dynamic channel as $\text{Pr}_{\rm fail}=0.4$.
 
\begin{table}[!htbp]
	\centering
	\caption{Reference values used in Simulations.}
	\begin{tabular}{lc}
		\hline
		\bf Parameter & \bf Adopted Value \\
		\hline
		Samples & $N_{s}=12$ \\
		$\text{MCS Trials}$ & $5000$ \\
		$\text{SUs}$ & SU $\in \{6, 10, 20\}$ \\
		$\text{PUs}$ & $1$\\
		${\rm Pr}_{\text{fail}}$ & $0.4$\\
	    SNR Range & SNR$_{\textsc{su}} \in \lbrace -10,\, 5\rbrace$ [dB] \\
		$\text{Channels}$ & $\text{AWGN, Rayleigh}$ \\
		$\text{Network}$ & $\text{Fixed, Dynamic}$ \\
		\hline
	\end{tabular}
	\label{tab_ref_values1}
\end{table}
 
\subsection{Network Topology} \label{NT}

In this work, we consider three different topologies to the cognitive network. The distributed network topology is based on graph theory. The application of graph theory in network context for consensus spectrum sensing purpose has been described in the section \ref{GT}.

\subsubsection{Topology I - 6 SUs} \label{NT6}

This topology is based on \cite{Kailkhura17} and depicted previously in Fig. \ref{ex6}. The $6$ SUs cooperate each other until the consensus convergence. \ref{ex6} The associated adjacency matrix is defined in Eq.\eqref{matrix6b}.

\subsubsection{Topology II - 10 SUs}

This topology is based on \cite{Li10}, \cite{Mitola15} and \cite{Mitola11}. The $10$ SUs cooperate each other until the consensus convergence. The Fig. \ref{ex10} shows the network topology. 
\begin{figure}[!htbp]
	\centering
    \includegraphics[width=.48\textwidth]{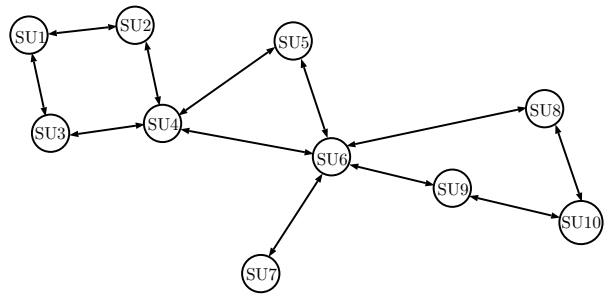}
	\caption{Decentralised Cooperative Scheme with $10$ SUs \cite{Mitola15}.}
	\label{ex10}
\end{figure}

\noindent As a consequence, the adjacent matrix in equation (\ref{matrix10}) defines the network topology represented by the graph of Fig. \ref{ex10}.

\begin{equation} 
{\bf G_{10}} = 
\small
\begin{bmatrix}
    0 & 1 & 1 & 0 & 0 & 0 & 0 & 0 & 0 & 0 \\
    1 & 0 & 0 & 1 & 0 & 0 & 0 & 0 & 0 & 0 \\
    1 & 0 & 0 & 1 & 0 & 0 & 0 & 0 & 0 & 0 \\
    0 & 1 & 1 & 0 & 1 & 1 & 0 & 0 & 0 & 0 \\
    0 & 0 & 0 & 1 & 0 & 1 & 0 & 0 & 0 & 0 \\
    0 & 0 & 0 & 1 & 1 & 0 & 1 & 1 & 1 & 0 \\
    0 & 0 & 0 & 0 & 0 & 1 & 0 & 0 & 0 & 0 \\
    0 & 0 & 0 & 0 & 0 & 1 & 0 & 0 & 0 & 1 \\
    0 & 0 & 0 & 0 & 0 & 1 & 0 & 0 & 0 & 1 \\
    0 & 0 & 0 & 0 & 0 & 0 & 0 & 1 & 1 & 0
\end{bmatrix}
\label{matrix10}
\end{equation}

\subsubsection{Topology III -- 20 SUs}
We create a new topology to characterise the performance of the DCSS methods in larger networks. The $20$ SUs cooperate each other until the cooperative SS consensus achieves convergence. Fig. \ref{ex20} depicts the graph for the network topology and the adjacent matrix ${\bf G}_{20}$ is straightforwardly defined in a similar way of the ${\bf G}_{10}$ in the {\it Topology II}. 

\begin{figure}[!htbp]
	\centering
    \includegraphics[width=.49\textwidth]{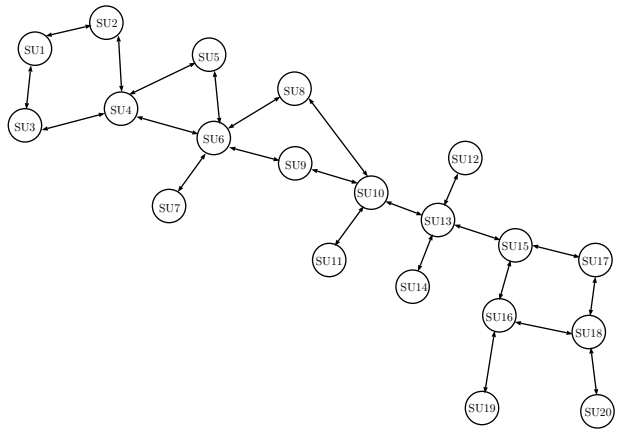}
	\caption{Decentralised cooperative scheme with $20$ SUs.}
	\label{ex20}
\end{figure}

\subsection{Parameters Values and Scenarios}
The two main parameters analysed in this work are the numerical {\it cooperative consensus convergence} (3C) and {\it receiver operator characteristics} (ROC). The goal of the numerical convergence analysis is determine and compare the number of iterations needed for the each the consensus SS technique achieves practical convergence. The parameter considered herein is the level of energy of each energy detector in dB. The cooperative consensus convergence is given when the energy difference $\Delta E$ among all the SUs output energy detected is $\Delta E\leq 1$ dB. The ROC analysis is the main figure-of-merit of analysis in the SS methods. The ROC is the relation of th probability of detection against the probability of false alarm. 

\subsection{Convergence}\label{sec:converg}
In this subsection we consider the numerical convergence as a figure-of-merit for analysis of the four consensus-based distributed spectrum sensing methods. The consensus methods are numerically compared considering the {different scenarios aiming at demonstrating the effectiveness} of the spectrum sensing methods. The results regarding the number of iterations for convergence is synthesised in Table \ref{tab:converg}. 
\begin{table}[!htbp]
\caption{Number of iterations for the DCSS method achieves convergence under $\Delta E \leq 1\,$ {[dB]}.}
\begin{center}
\begin{tabular}{cc|cccc}
	\hline
	\bf Scenario & \bf \#SUs &\bf AC& \bf WAC& \bf WAC-AE & \bf IWAC\\
	\hline\hline
	{\bf A}-AWGN & $6$  &  $4$  & $15$ & $5$ & $15$\\
   		(Fixed)           & $10$ & $4$ & $6$ & $9$ & $10 $ \\
	\hline
	{\bf B}-AWGN & $10$ & $4$ & $6$ & $9$ & $10$ \\
  		(Mobile)         & $20$ & $22$ & $25$ & $30$ & $31$ \\
	\hline \hline
	{\bf C}-Rayleigh & $6$ & $15$ & $19$ & $35$ & $34$ \\
         (Fixed)       & $10$ & $19$ & $11$ & $18$ & $27$ \\
	\hline
	{\bf D}-Rayleigh & $10$ & $19$ & $11$ & $18$ & $27$ \\
     	(Mobile)        & $20$ & $42 $ & $48$ & $>50$ & $>50$ \\
	\hline
\end{tabular}
\end{center}
\label{tab:converg}
\end{table}

For scenario A, the network with $10$ SUs needs less average number of iterations to reach the convergence criterion $\Delta E \leq 1$ [dB], compared to the network with $6$ SUs, due to the higher availability of connections among the SU neighbours.  On the average, the AC method needs less number of iterations than the WAC, WAC-AE and IWAC methods to achieve convergence in almost all scenarios, including AWGN $\times$ Rayleigh, fixed $\times$ mobile channels, and a low-medium $\times$ a high number of cooperative SUs.

{In most cases}, the IWAC method requires a higher number of iterations to achieve $\Delta E$-based convergence, while the WAC-AE method operating under dynamic/mobile channels needs approximately the same number of iterations compared to the IWAC method, but yet higher than AC and WAC methods. Moreover, as expected, in the Rayleigh channel scenarios, all methods require a higher number of iteration to achieve convergence due to the channel characteristics. Notice that in the analysed numerical simulations, we have averaged on $500$ channel realisations: the Rayleigh channel coefficients, as well as SU localisation (reflecting different SNRs$_{\textsc{su}}$) have been taken randomly and deployed to characterise the SS detectors' convergence.

Fig. \ref{fig:converg3} depicts convergence behaviour for the four AC detectors in the case of $10$ SUs operating under dynamic AWGN channels, while Fig. \ref{fig:converg5} reveals the convergence trend for the case of $10$ cooperative SUs in a fixed network under Rayleigh channels.

\begin{figure}[!htbp]
\centering
\small
\includegraphics[width=.495\textwidth]{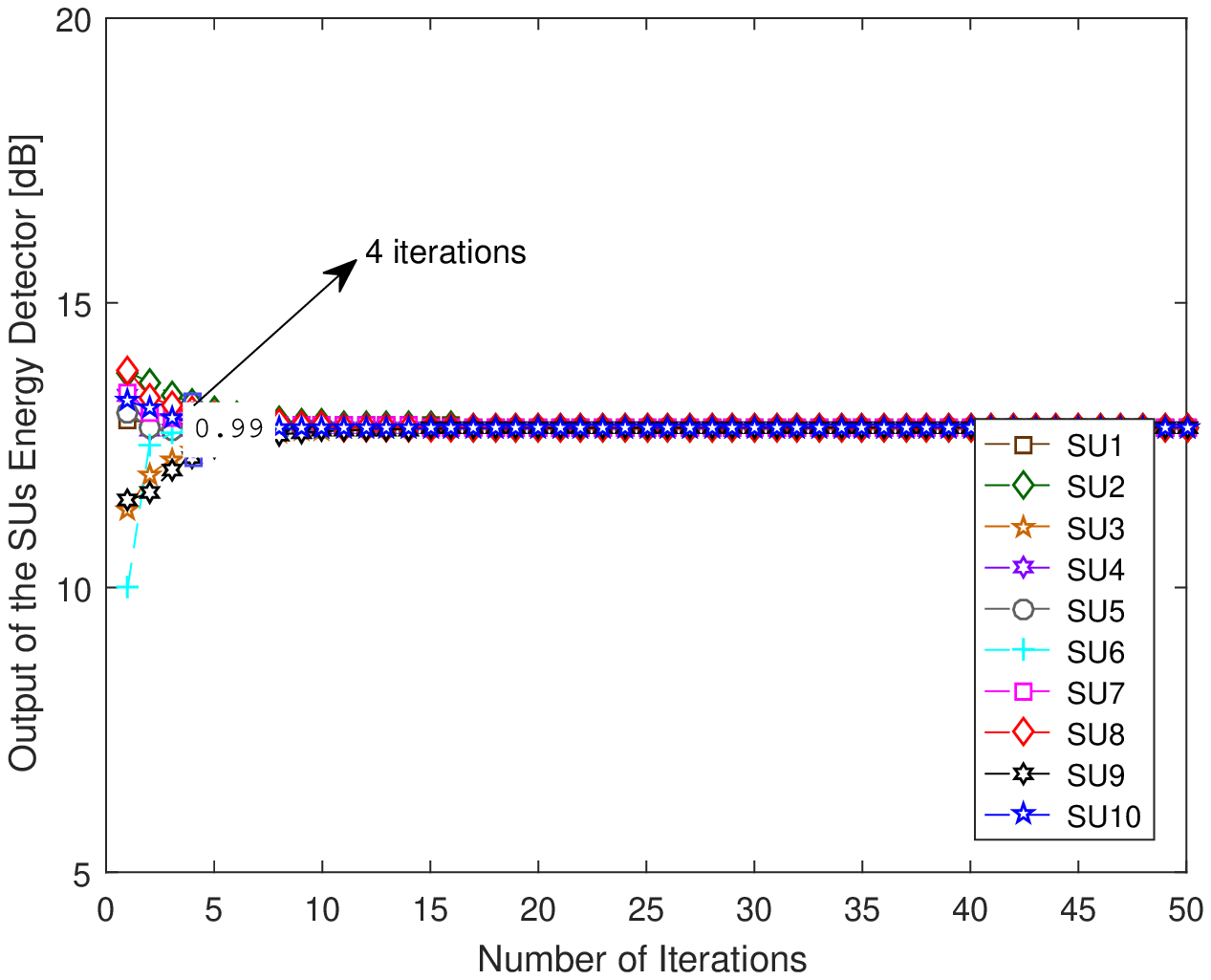}
\includegraphics[width=.495\textwidth]{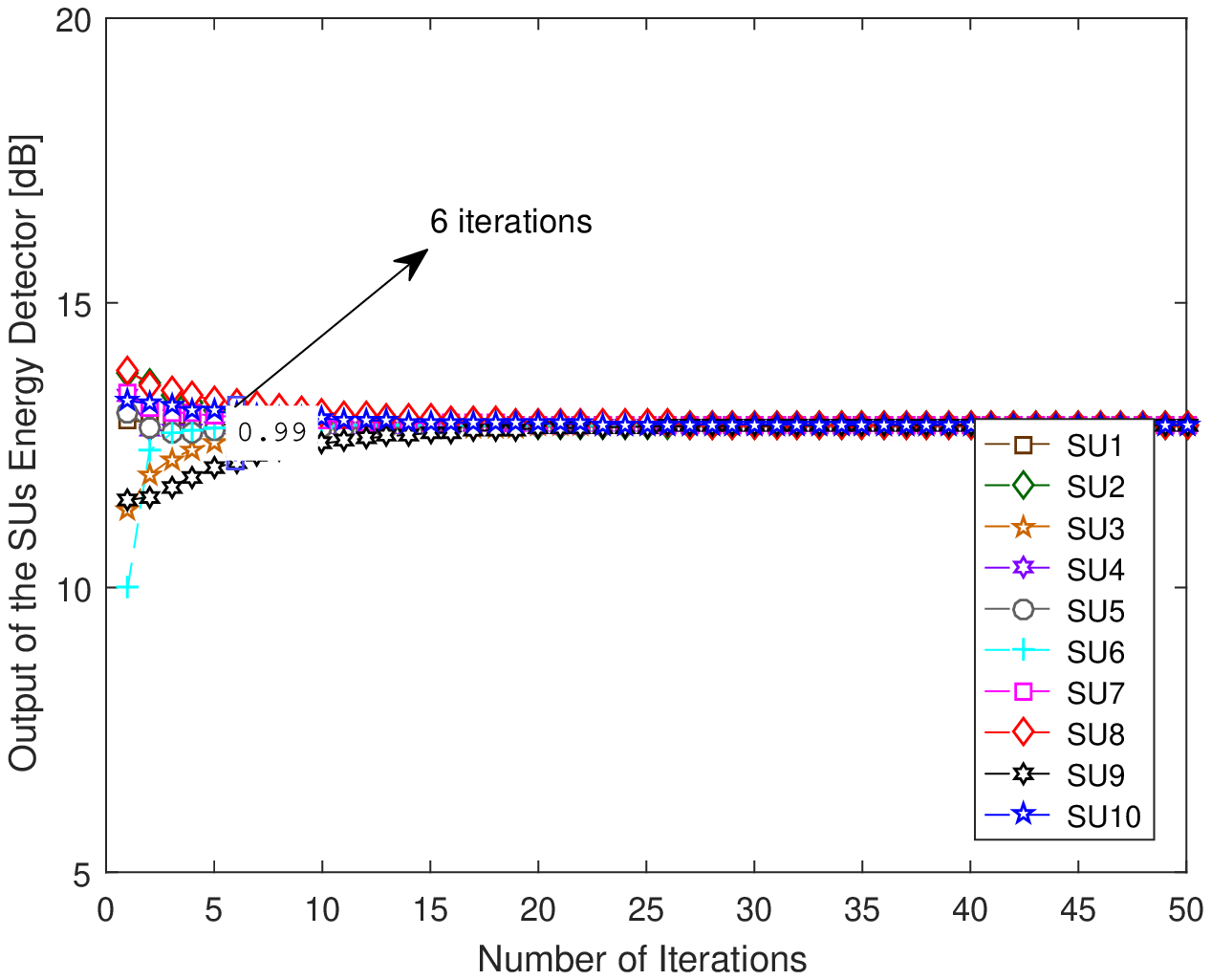}\\
    a) AC \hspace{7cm}  b) WAC\\
\includegraphics[width=.495\textwidth]{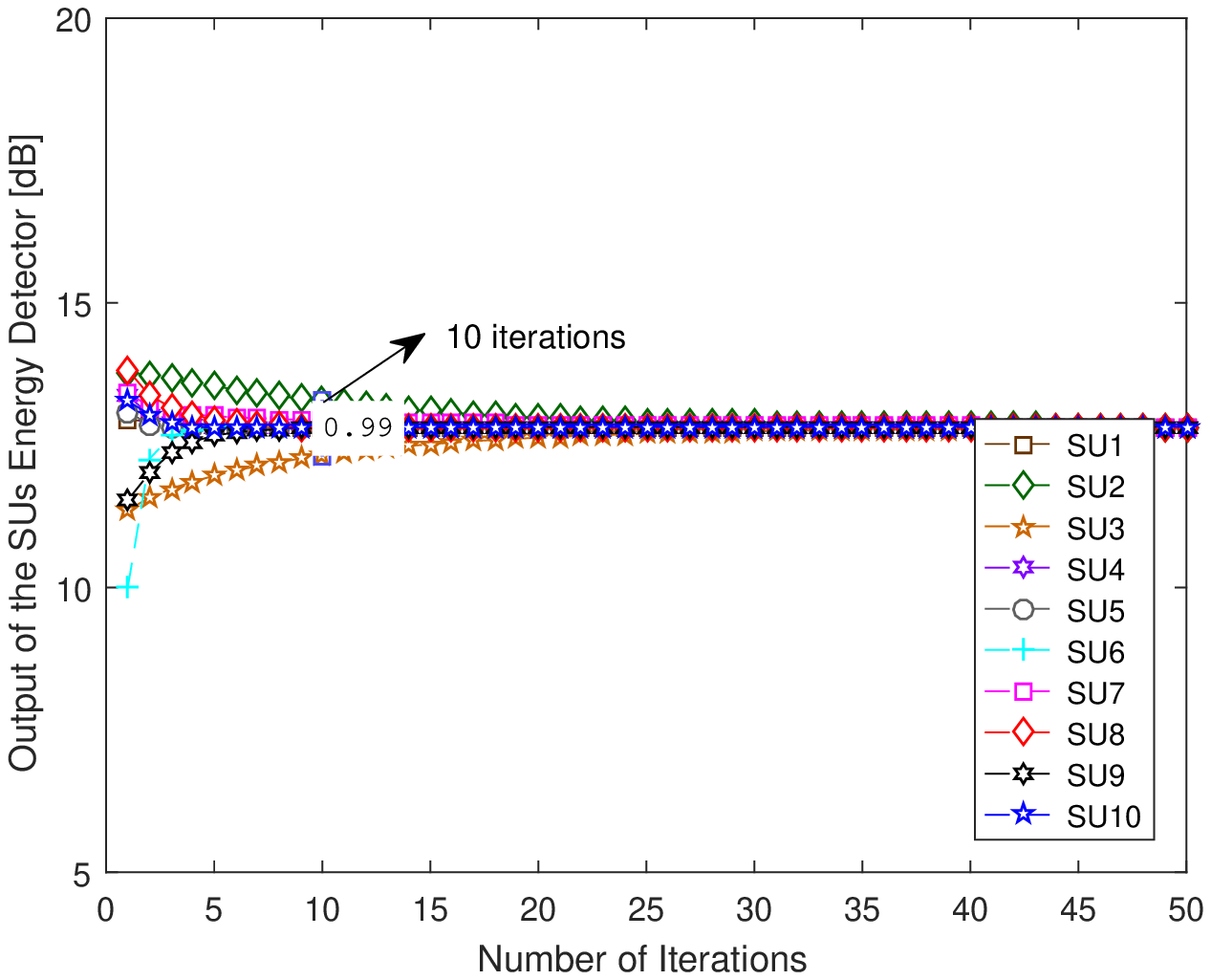}
 \includegraphics[width=.495\textwidth]{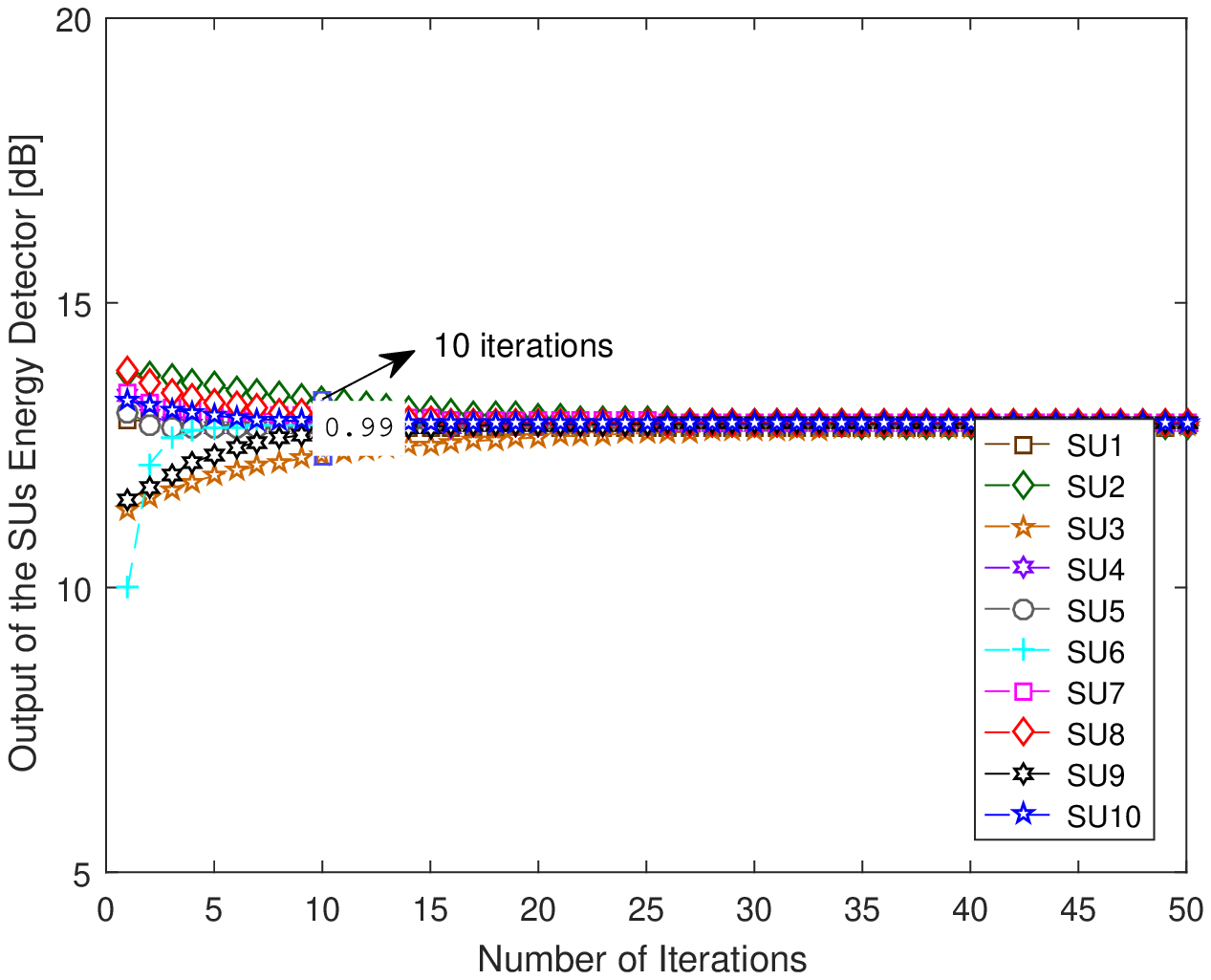}
	c) WAC-AE \hspace{7cm} d) proposed IWAC
\caption{Convergence for the different DCSS under AC rules considering $10$ SUs and Dynamic AWGN Channel.}
\label{fig:converg3}
\end{figure}

\begin{figure}[!htbp]
\centering
\small
\includegraphics[width=.495\textwidth]{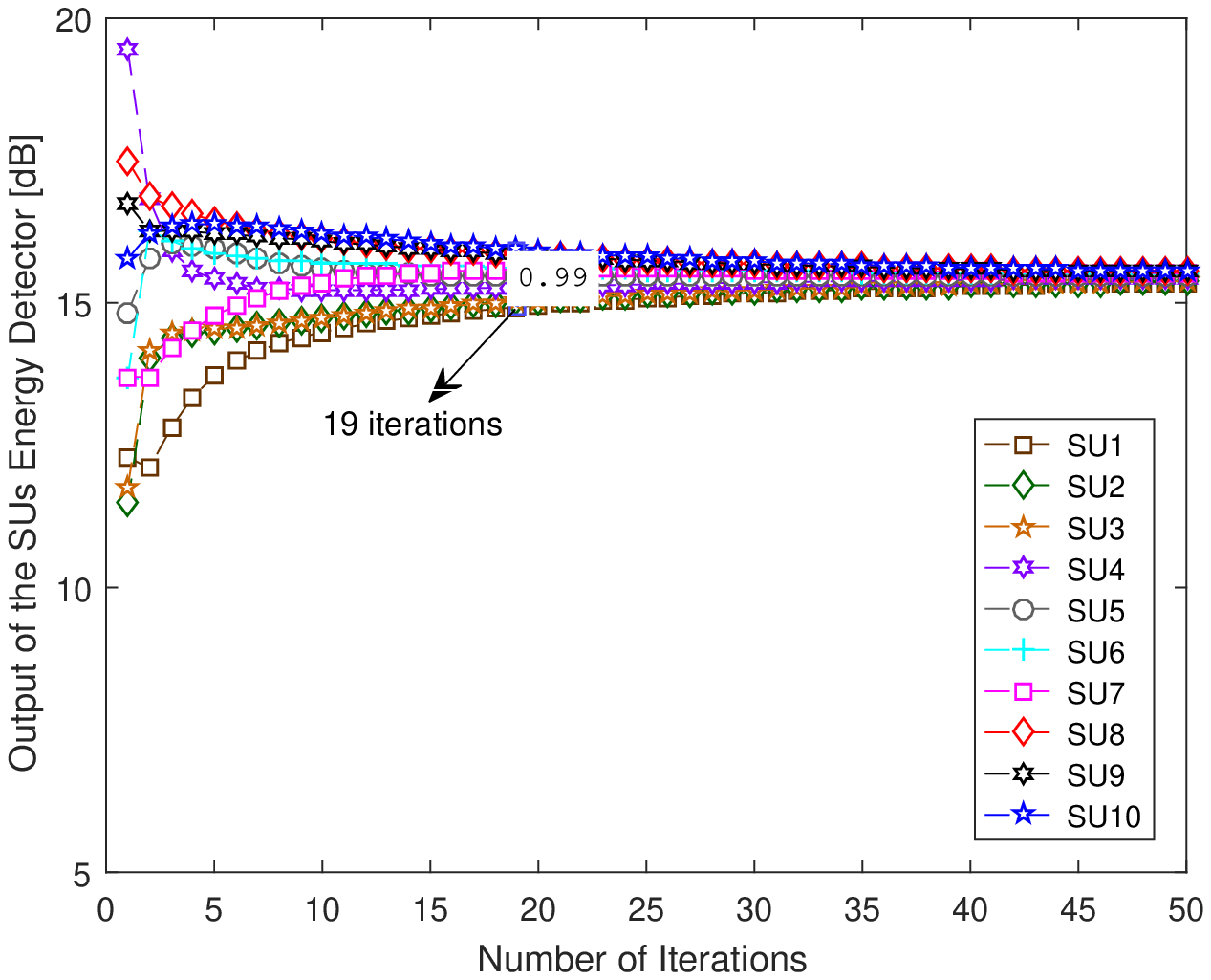}
\includegraphics[width=.495\textwidth]{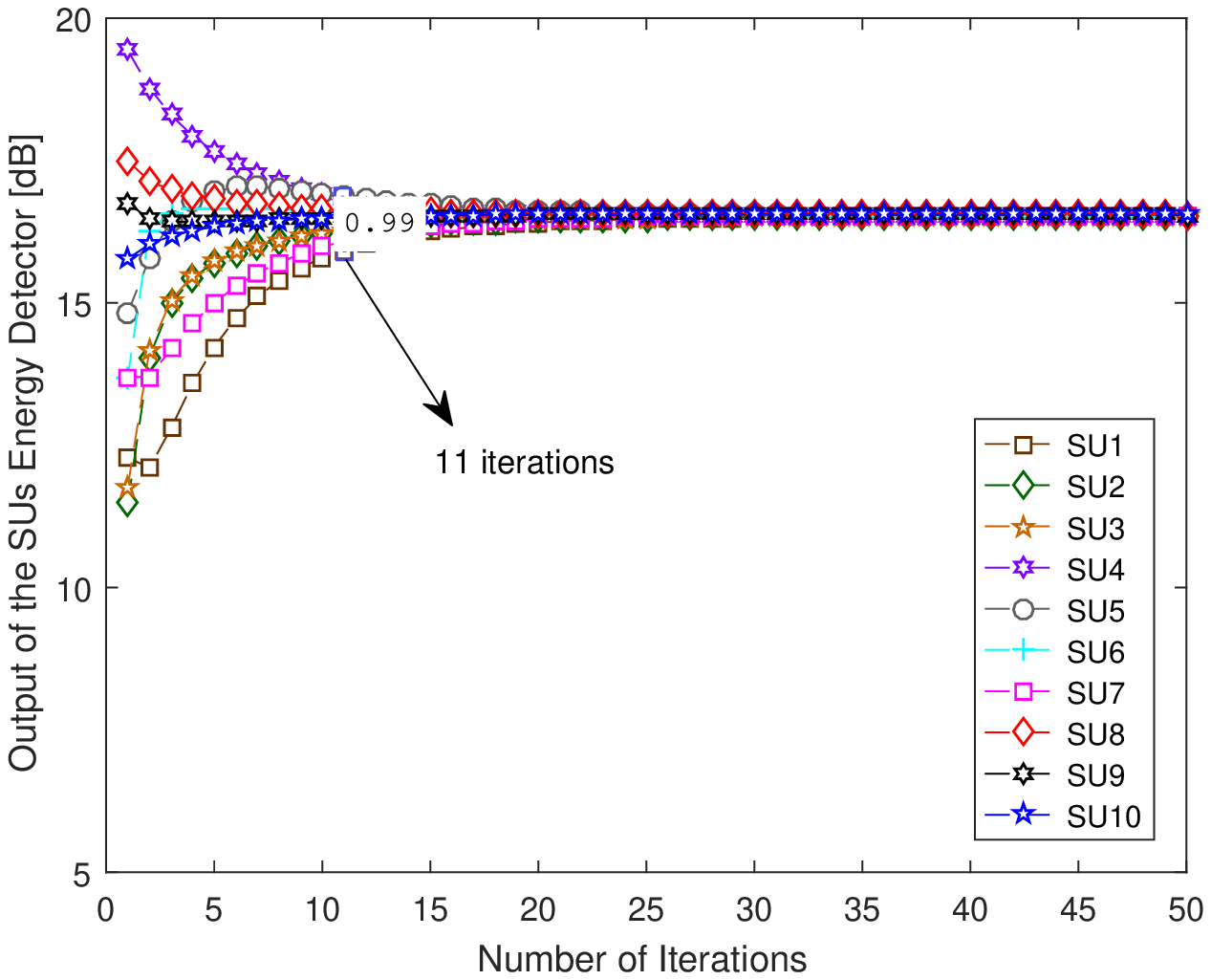}\\
    a) AC \hspace{7cm}  b) WAC\\
\includegraphics[width=.495\textwidth]{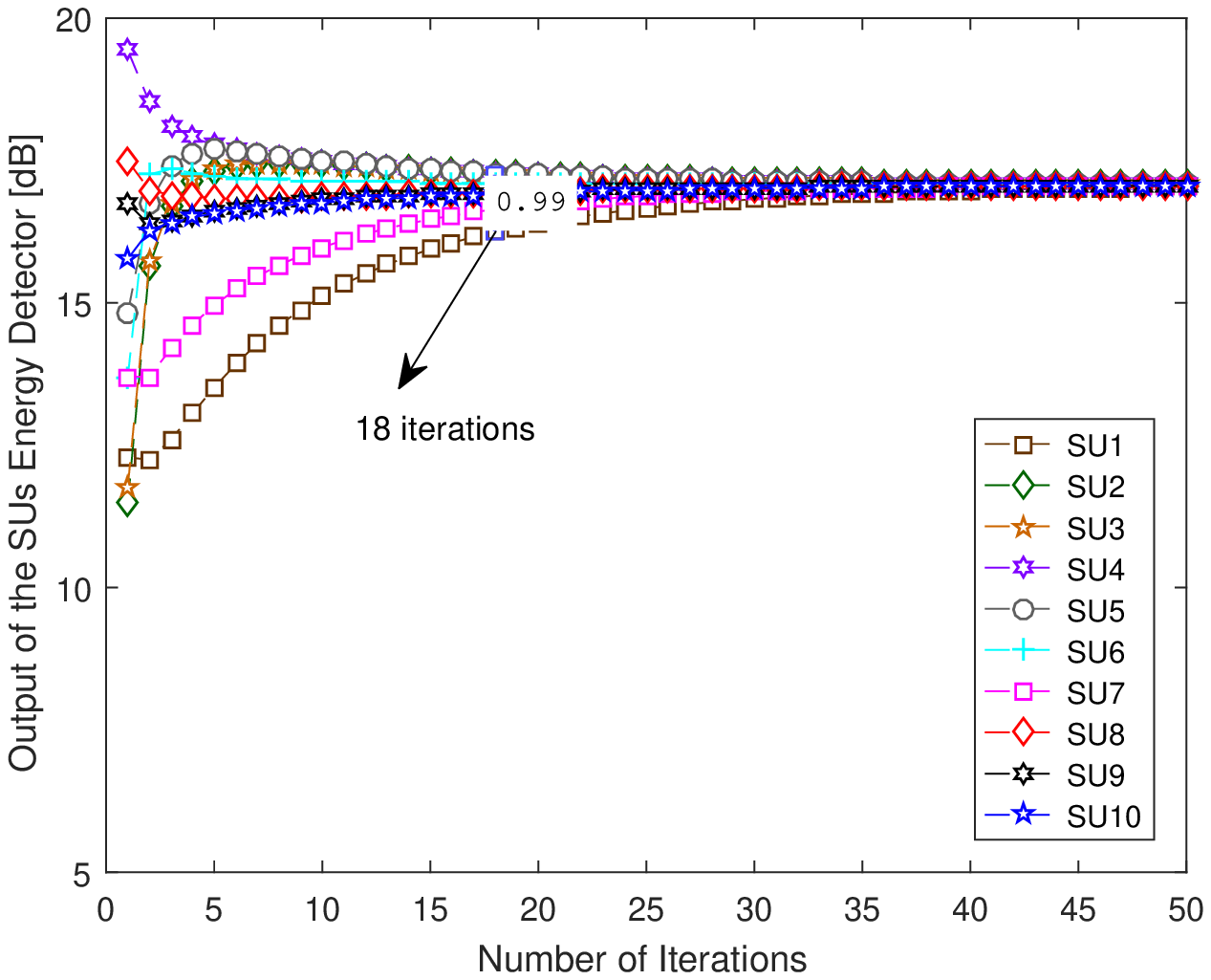}
 \includegraphics[width=.495\textwidth]{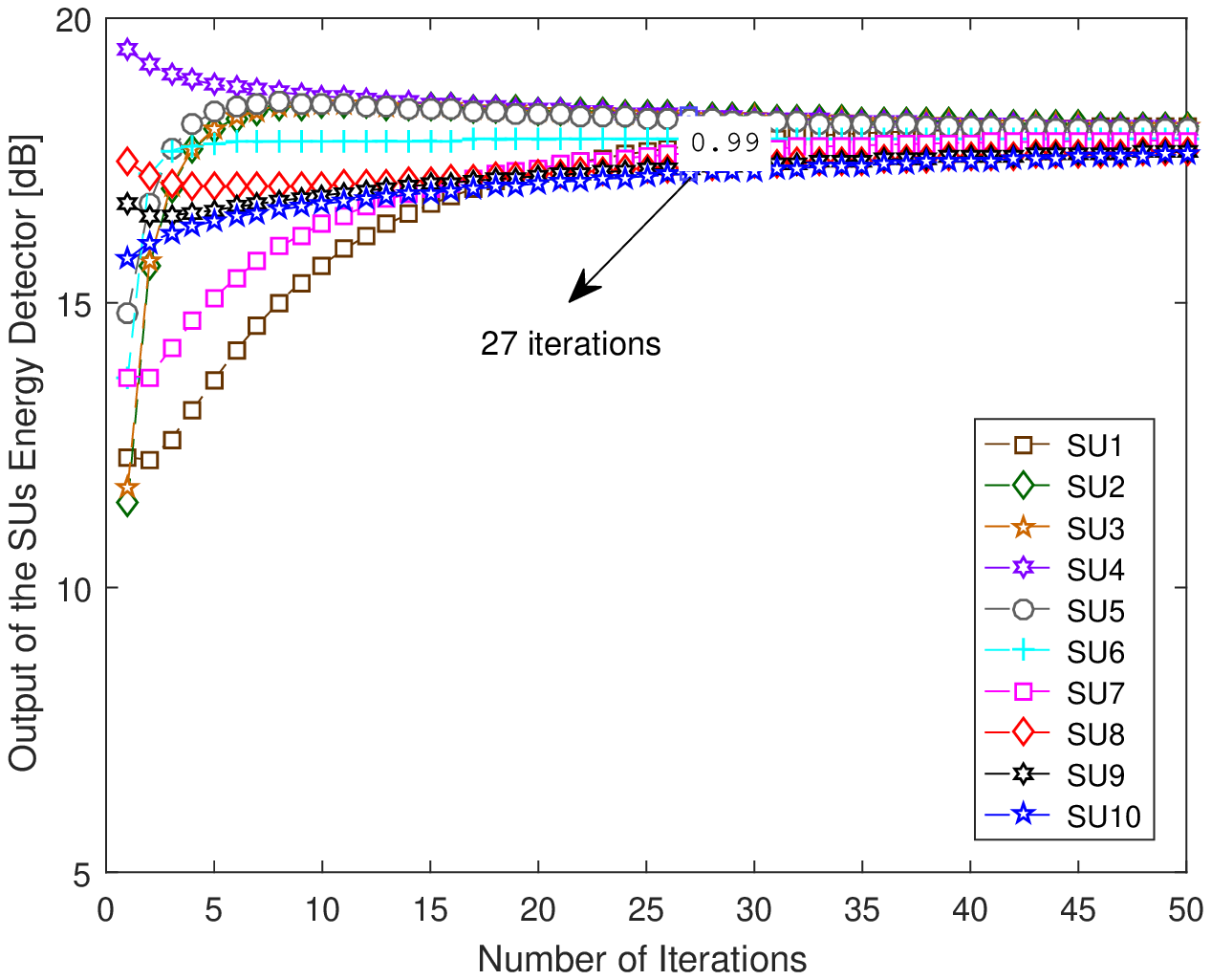}
	c) WAC-AE \hspace{7cm} d) proposed IWAC
\caption{Convergence for the different DCSS rules considering $10$ SUs under Fixed Rayleigh channel.}
\label{fig:converg5}
\end{figure}

\subsection{ROC}
 The global ROC  for the various spectrum {sensing methods is} numerically compared considering {different scenarios (A, B, C and D) aiming at demonstrating the effectiveness} of the proposed cooperative IWAC method under both AWGN and NLOS-Rayleigh channels. Indeed, Fig. \ref{fig:GlobalROC} depicts the ROC for several classical as well the proposed IWAC and WAC-AE DCSS methods, considering $6$, $10$ and $20$ SUs, AWGN Channel, Fixed and Dynamic Networks.

 For the $6$ SUs the WAC, as well as the proposed WAC-AE method have similar performance and can be compared to the MRC rule, which represents the optimum centralised SS performance. The proposed IWAC method presents a slight degradation compared to the WAC and WAC-AE methods, but keeps better performance compared to the AC method, which has similar performance to the EGC rule. On the other hand, the classical hard combining rules result in poor performance compared to the soft combining rule. Among all classical rules,  the OR rule has the best performance while the AND rule presents the worse performance. A similar conclusion can be obtained for $10$ and $20$ SUs (see Fig.  \ref{fig:GlobalROC}.b, \ref{fig:GlobalROC}.c and \ref{fig:GlobalROC}.d). Moreover, the mobility of network does not affect substantially the ROC performance of all spectrum sensing techniques operating under AWGN channels.

The ROC behaviour for the nine spectrum sensing rules operating under Rayleigh channels and  $6$, $10$ and $20$ fixed and dynamic SUs is depicted in Fig. \ref{fig:GlobalROC2}. Again, for $6$ SUs the IWAC, WAC-AE and WAC methods demonstrate similar performance when compared to the optimum performance (MRC rule). The AC method has similar performance to the EGC rule, and for this scenario, it results in a similar performance of the MRC and WAC methods. Interesting, one can conclude that in severe Rayleigh fading channels scenarios the OR rule results in suitable performance while the AND rule performances worse. Similar conclusion can be obtained for a different number of cooperative SUs. Finally, the mobility of network does not affect the ROC performance substantially. Note that the suitable ROC performance achieved for all rules, except AND rule, under Rayleigh channels could be attained due to a higher range of   SNR$_{\textsc{su}}  \in  \lbrace -2, \,\, 5\rbrace$ [dB] when compared with the SNR range adopted in AWGN scenarios.

\begin{figure}[!htbp]
\centering
\small
\includegraphics[width=.495\textwidth]{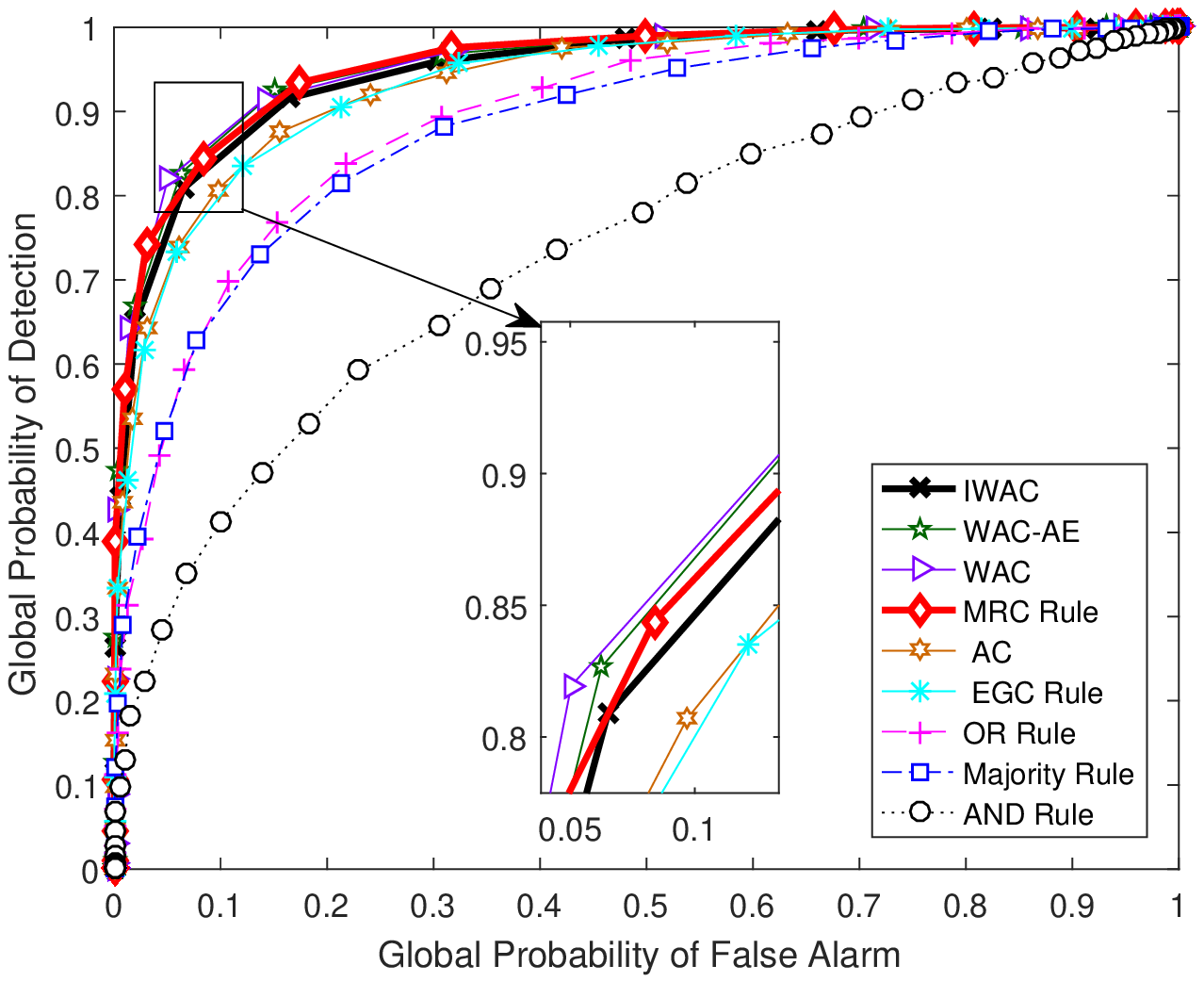}
\includegraphics[width=.495\textwidth]{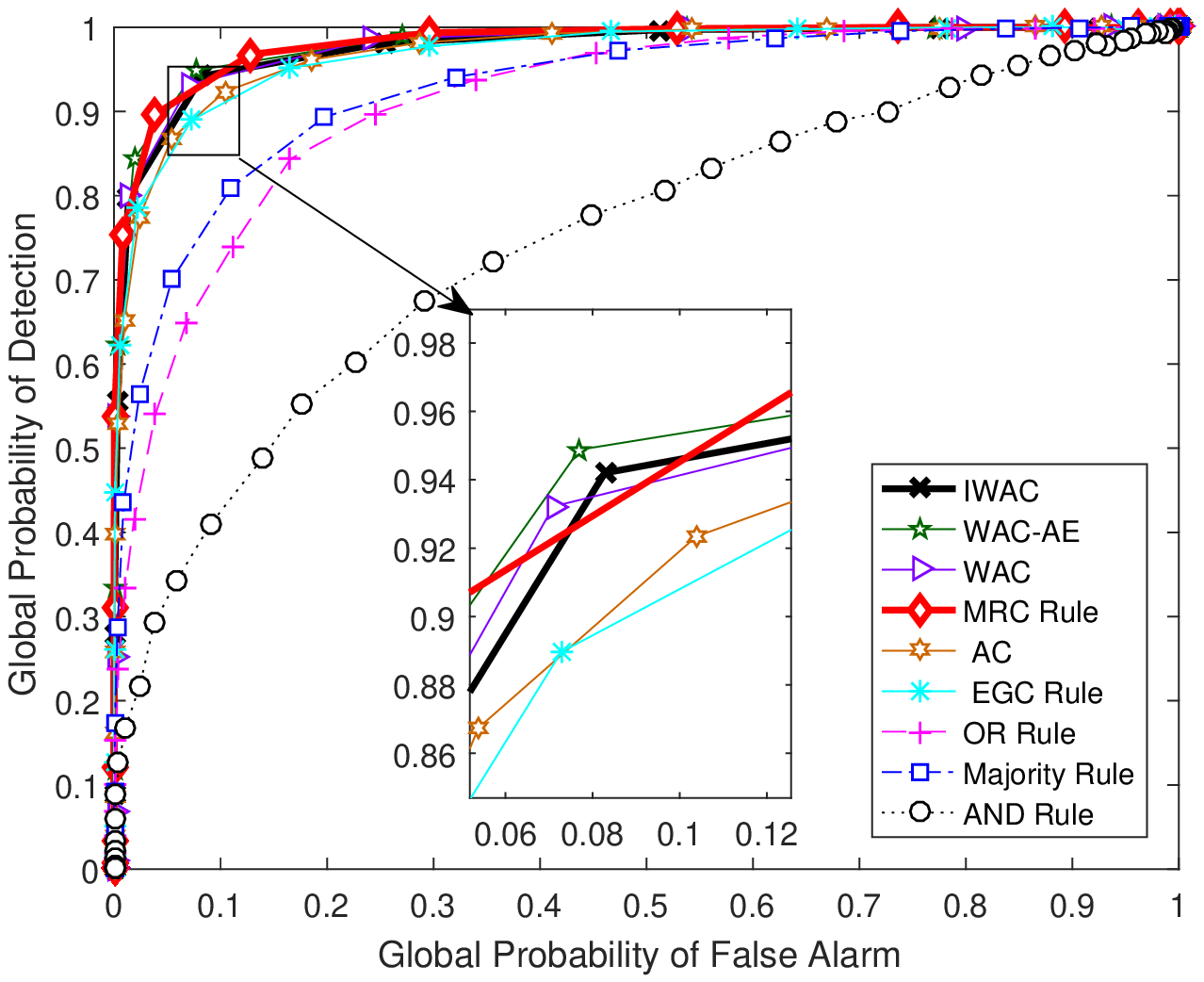}\\
	a) {\bf Fixed} Network, $6$ SUs \hspace{5cm}  b) {\bf Fixed} Network, $10$ SUs \\
 \includegraphics[width=.495\textwidth]{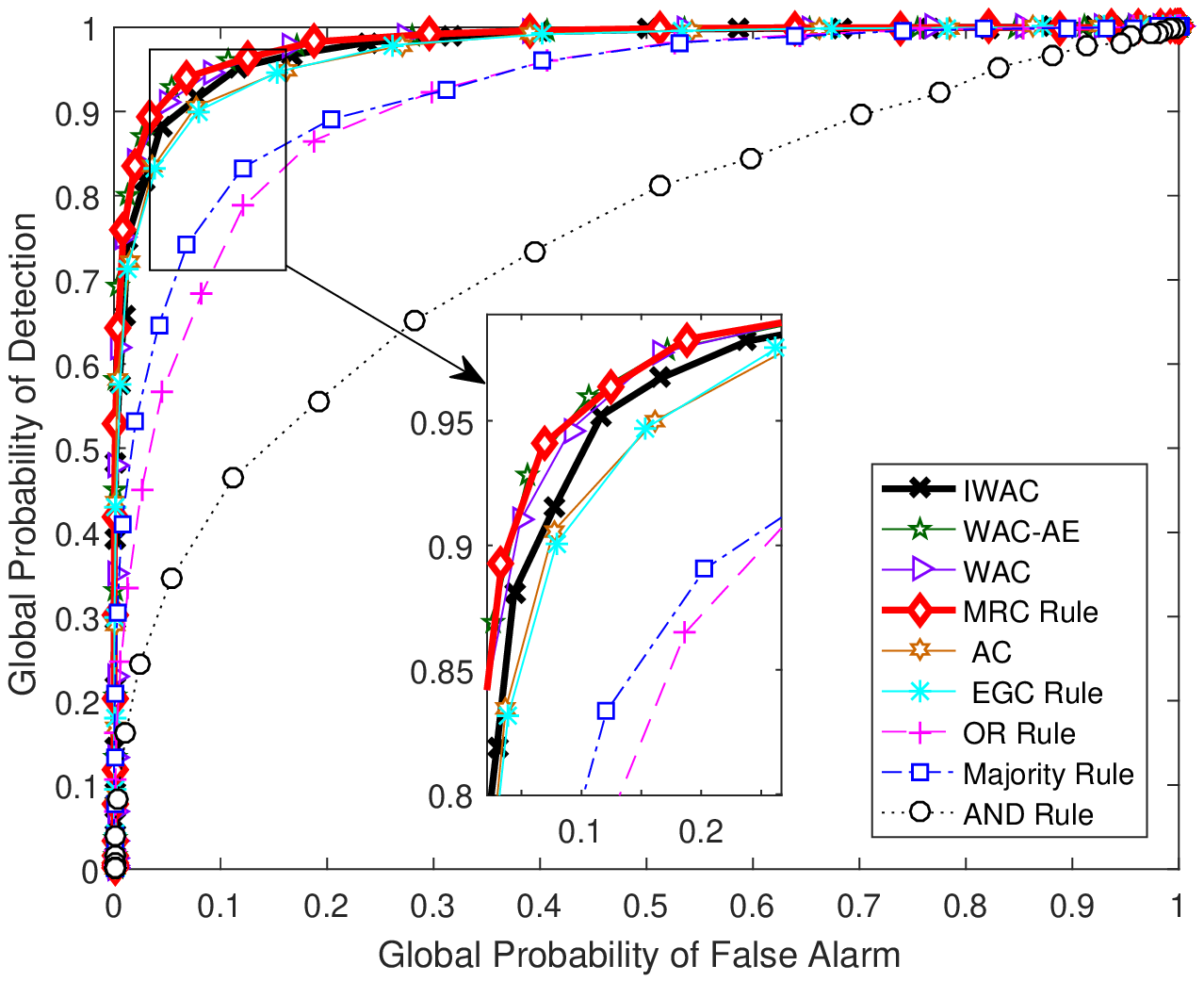}
 \includegraphics[width=.495\textwidth]{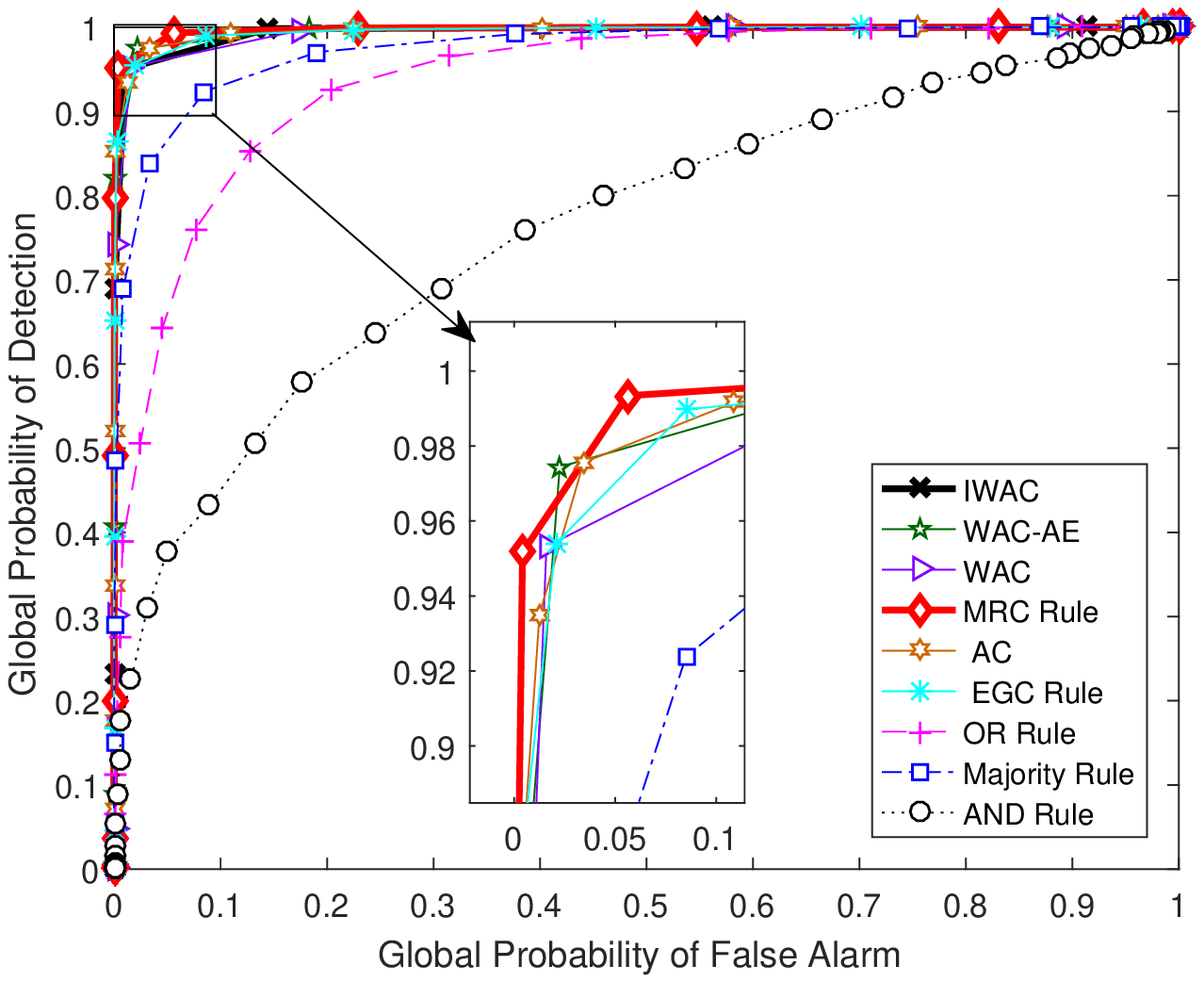}
 	c) {\bf Mobile} Network, $10$ SUs \hspace{5cm}  d) {\bf Mobile} Network, $20$ SUs \\
 \caption{Global ROC for several DCSS methods operating with  $6$, $10$ and $20$ SUs for Fixed and Dynamic Networks in {\bf AWGN channels}.}\label{fig:GlobalROC}
\end{figure}

\begin{figure}[!htbp]
\centering
\small
\includegraphics[width=.495\textwidth]{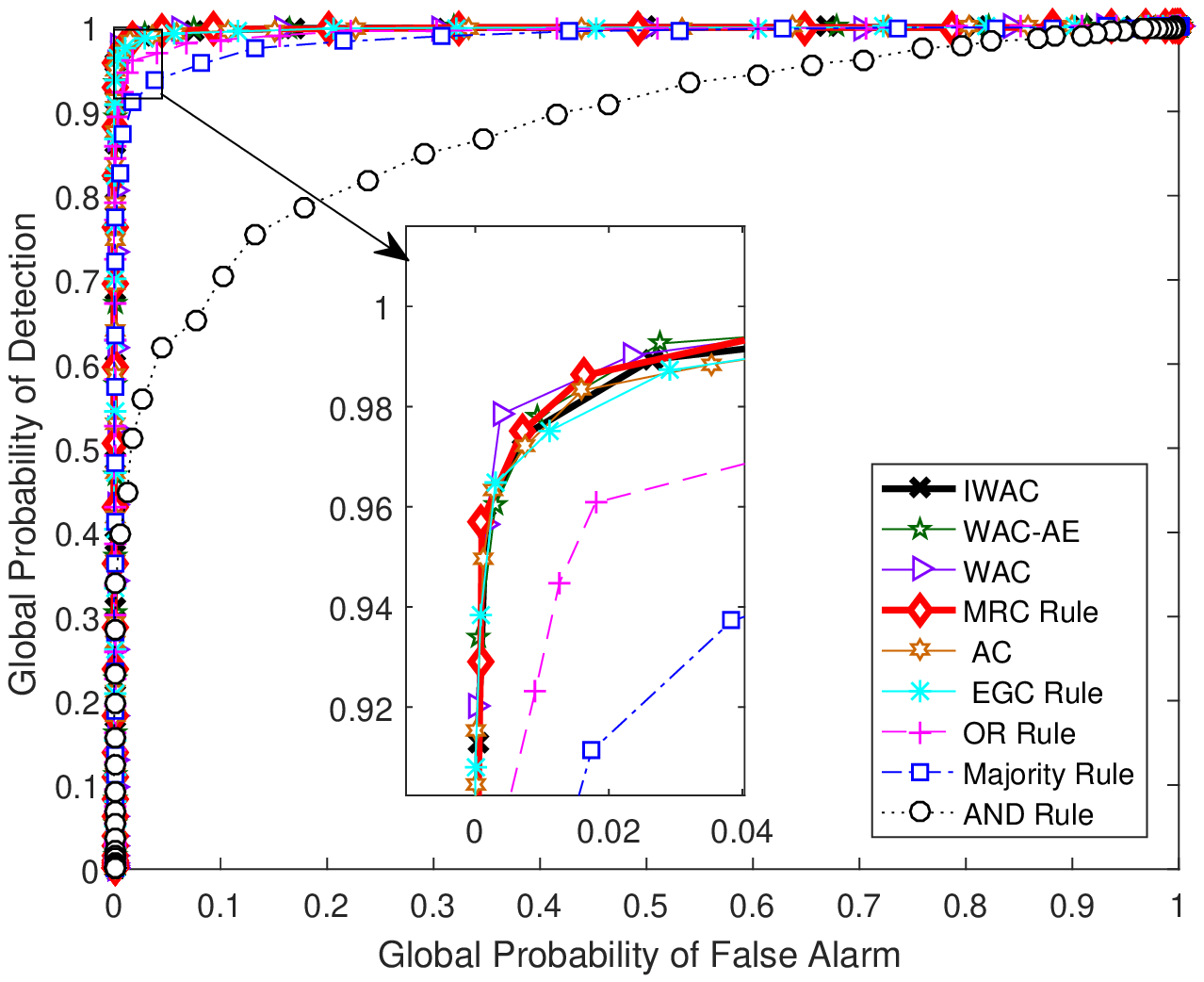}
\includegraphics[width=.495\textwidth]{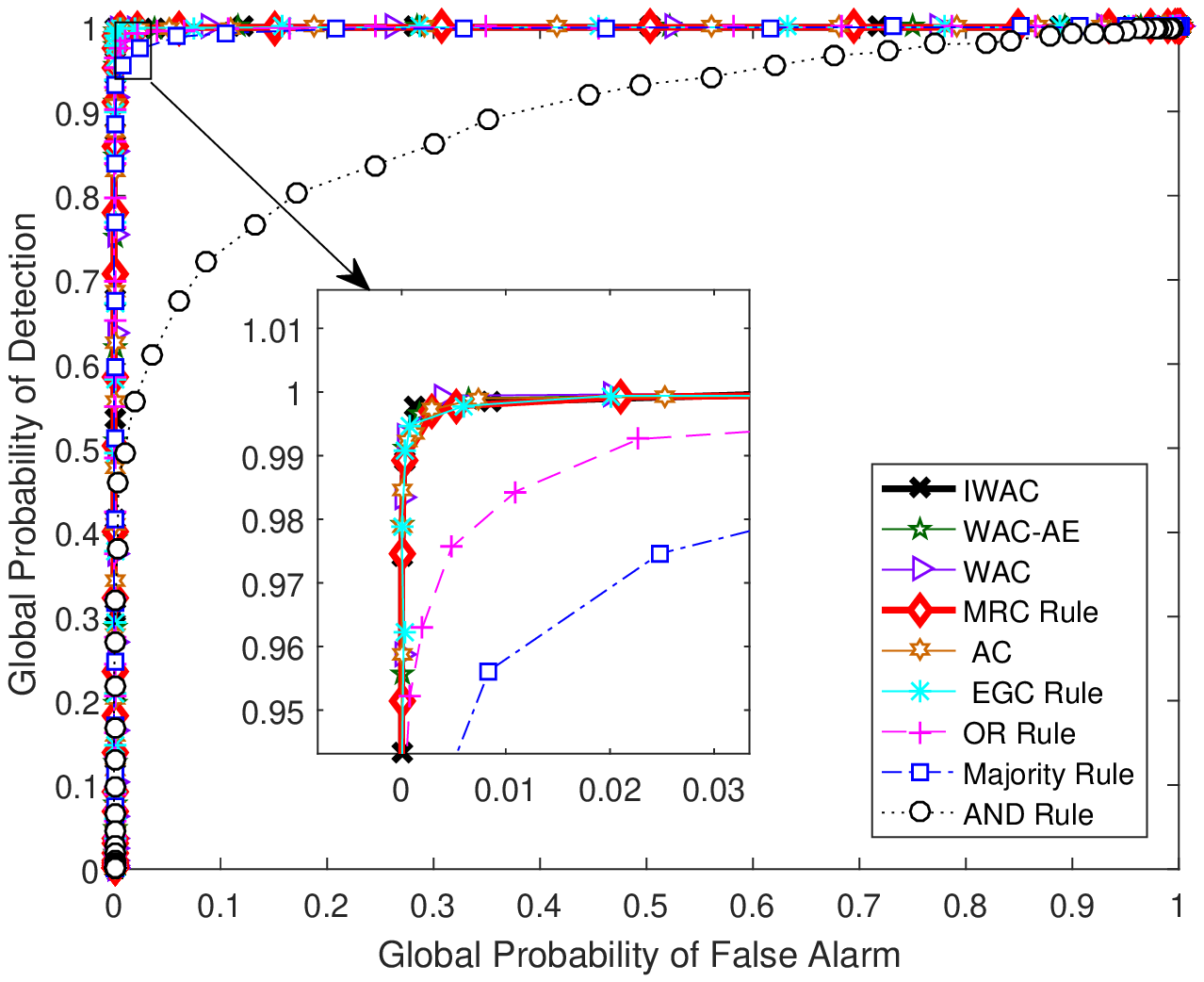}\\
	a) {\bf Fixed} Network, $6$ SUs \hspace{5cm}  b) {\bf Fixed} Network, $10$ SUs \\
 \includegraphics[width=.495\textwidth]{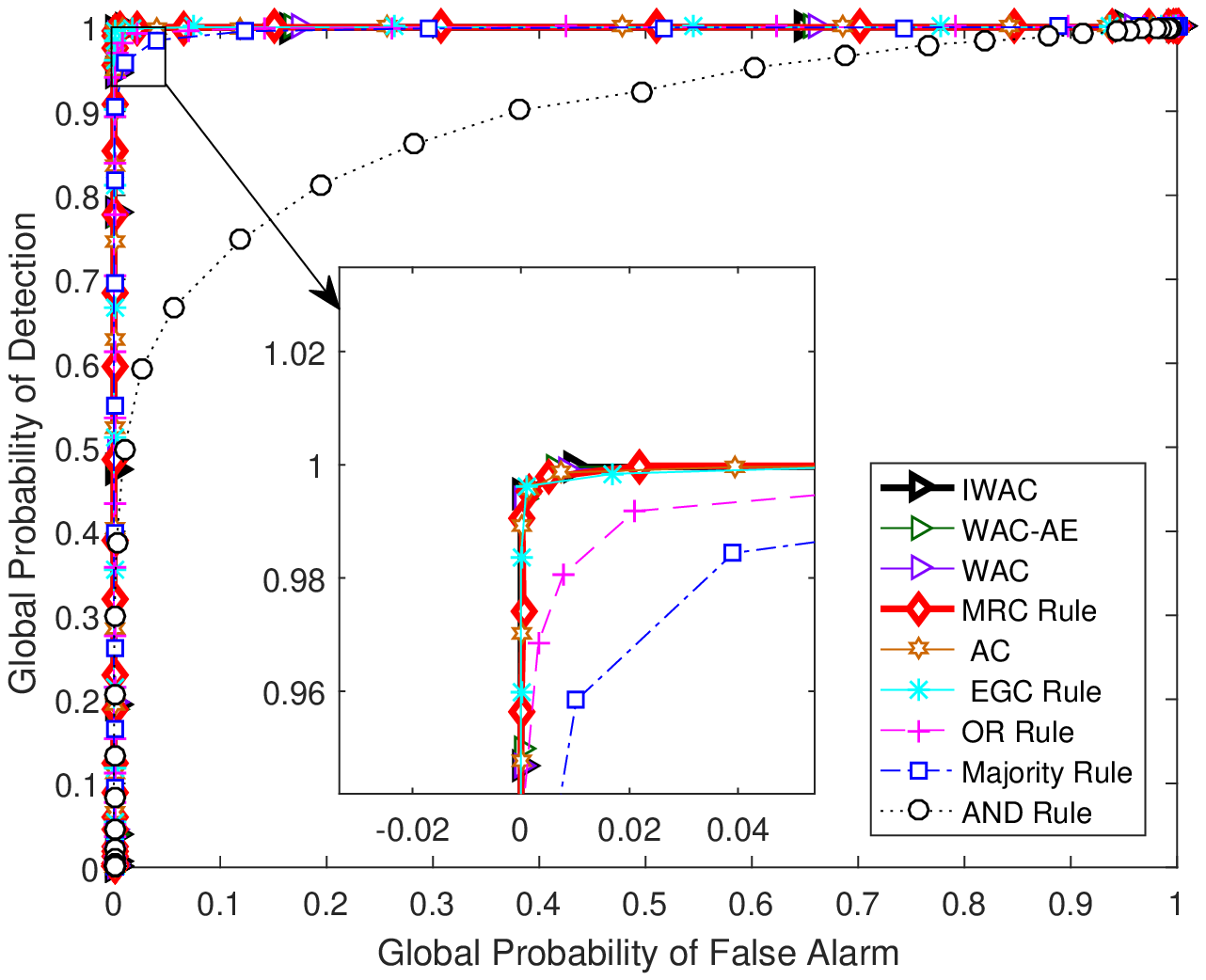}
 \includegraphics[width=.495\textwidth]{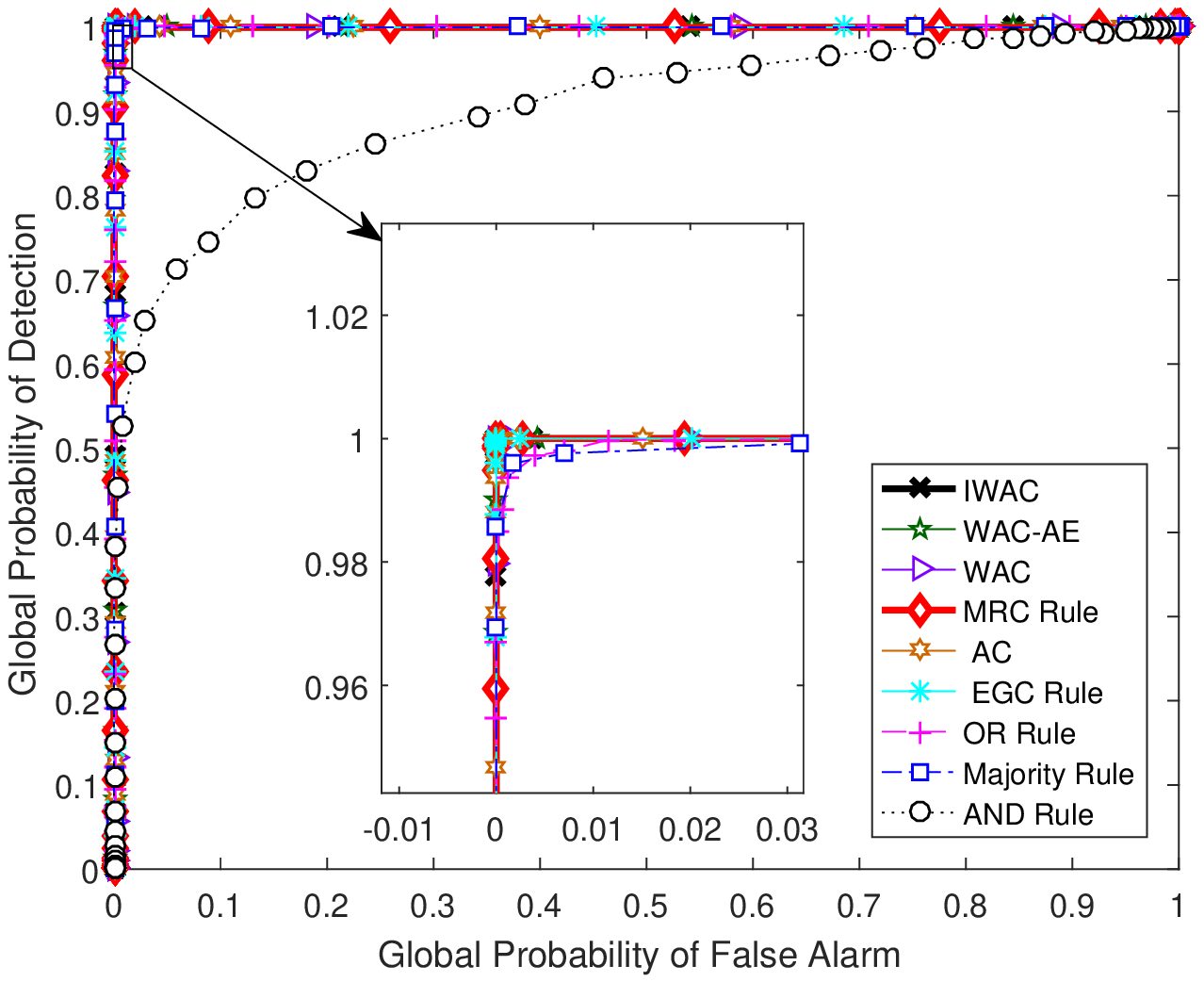}
 	c) {\bf Mobile} Network, $10$ SUs \hspace{5cm}  d) {\bf Mobile} Network, $20$ SUs \\
 \caption{Global ROC for several DCSS methods, $6$, $10$  and $20$ SUs for Fixed and Dynamic Networks operating under {\bf Rayleigh channels}.}\label{fig:GlobalROC2}
\end{figure}

\subsubsection{Analytical versus Simulated ROC}\label{sec:ROC}
Fig. \ref{fig:LocalROC} demonstrates the local (distributed) ROC  for the proposed IWAC-DCSS method considering only the Scenario A ($6$ and $10$ SUs in an AWGN fixed channel). The analytical expression for the ROC of each SU inspired in \eqref{eq:Pcentralized}, but considering the local decision, is  compared with the numerical Monte-Carlo simulation results. The analytical performance considering a fail probability at the  $i$-th SU,  ${\rm P}_{f}^{i}$ can be described adapting the eq. \eqref{eq:Pcentralized} to distributed IWAC soft CSS decision:
\begin{equation} \label{eq:IWAC_analyticalROC1}
	{\rm P}_{d}^{i} =  Q\left(\frac{Q^{-1}({\rm P}_{f}^{i})\sqrt{\text{var}(x_{i}|\mathcal{H}_{0})} - \mathbb{E}(x_{i}|\mathcal{H}_{1}) + \mathbb{E}(x_{i}|\mathcal{H}_{0})}{\sqrt{\text{var}(x_{i}|\mathcal{H}_{1})}}\right), 
\end{equation}
where 
\begin{equation}\label{eq:IWAC_analyticalROC2}
	\mathbb{E}(x_{i}|\mathcal{H}_{0,1}) =
  \begin{cases}
   (\prod_{\ell = 1}^{k} {\bf P}_{\ell_{\textsc{iwac}}} \mathbb{E}({\bf x}(0)|\mathcal{H}_{0}) )_{i}& \quad ,\mathcal{H}_{0}\\
     (\prod_{\ell = 1}^{k} {\bf P}_{\ell_{\textsc{iwac}}} \mathbb{E}({\bf x}(0)|\mathcal{H}_{1}) )_{i} &\quad ,\mathcal{H}_{1}\\
  \end{cases}
\end{equation}

\begin{equation}\label{eq:IWAC_analyticalROC3}
 \text{var}(x_{i}|\mathcal{H}_{0,1}) =
  \begin{cases}
   (\prod_{\ell = 1}^{k} {\bf P}_{\ell_{\textsc{iwac}}} \text{cov}({\bf x}(0)|\mathcal{H}_{0})  \prod_{\ell = 1}^{k} {\bf P}_{\ell_{\textsc{iwac}}})_{ii} & \quad ,\mathcal{H}_{0}\\
     (\prod_{\ell = 1}^{k} {\bf P}_{\ell_{\textsc{iwac}}} \text{cov}({\bf x}(0)|\mathcal{H}_{1}) \prod_{\ell = 1}^{k} {\bf P}_{\ell_{\textsc{iwac}}})_{ii} &\quad ,\mathcal{H}_{1}\\
  \end{cases}
\end{equation}
where $\text{cov}({\bf x}) = \mathbb{E}[({\bf x} - \mathbb{E}({\bf x}))({\bf x} - \mathbb{E}({\bf x}))^{T}]$ is the covariance matrix of the vector ${\bf x}$.

Indeed, for scenario A, Fig. \ref{fig:LocalROC} demonstrates suitable fitting among  the Monte-Carlo simulated results and the analytical expression, evidencing that the set of eqs \eqref{eq:IWAC_analyticalROC1}-\eqref{eq:IWAC_analyticalROC3} is a valid analytical description to characterize the IWAC ROC performance.

\begin{figure}[!htbp]
\centering
\small
\includegraphics[width=.495\textwidth]{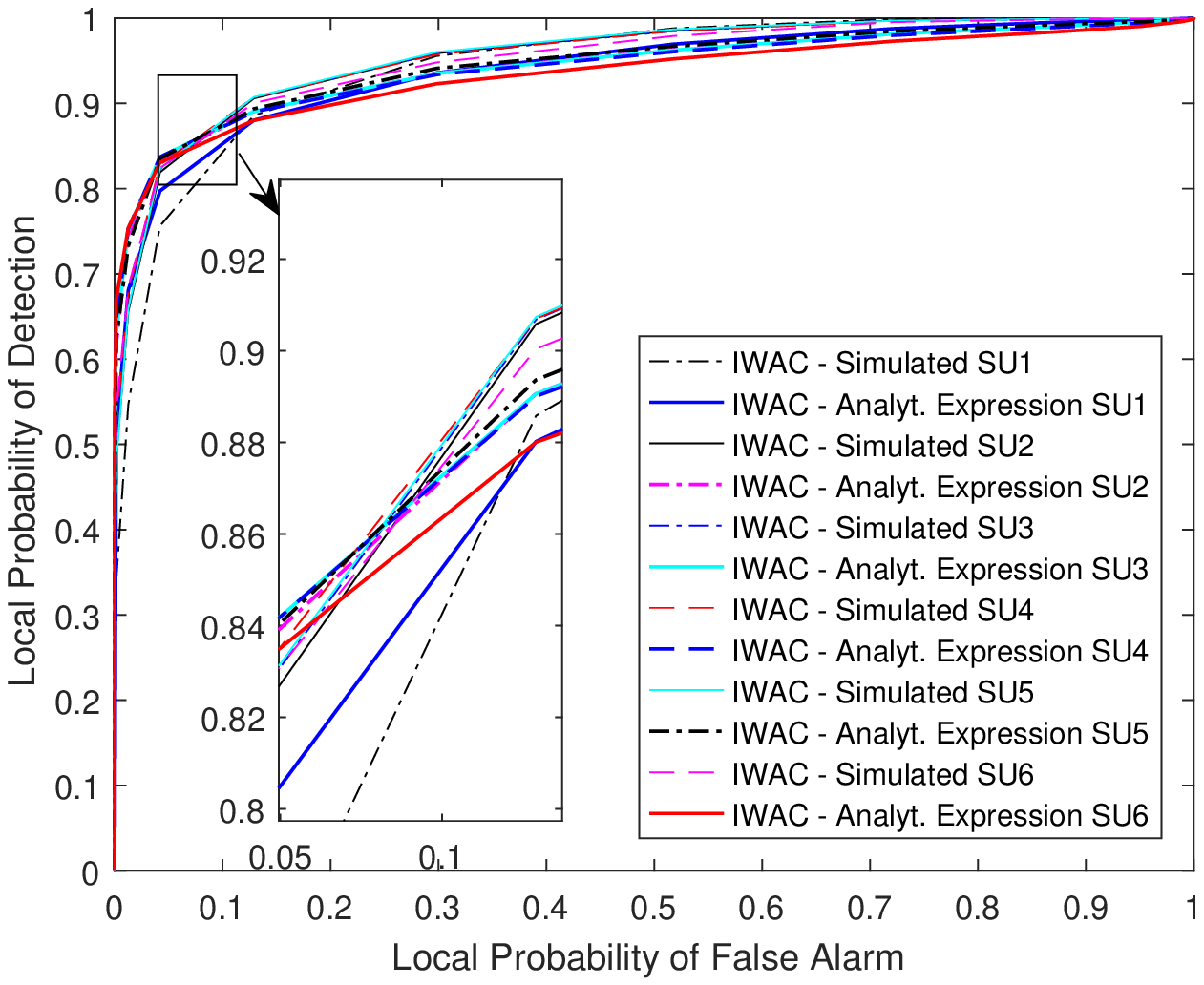}
\includegraphics[width=.495\textwidth]{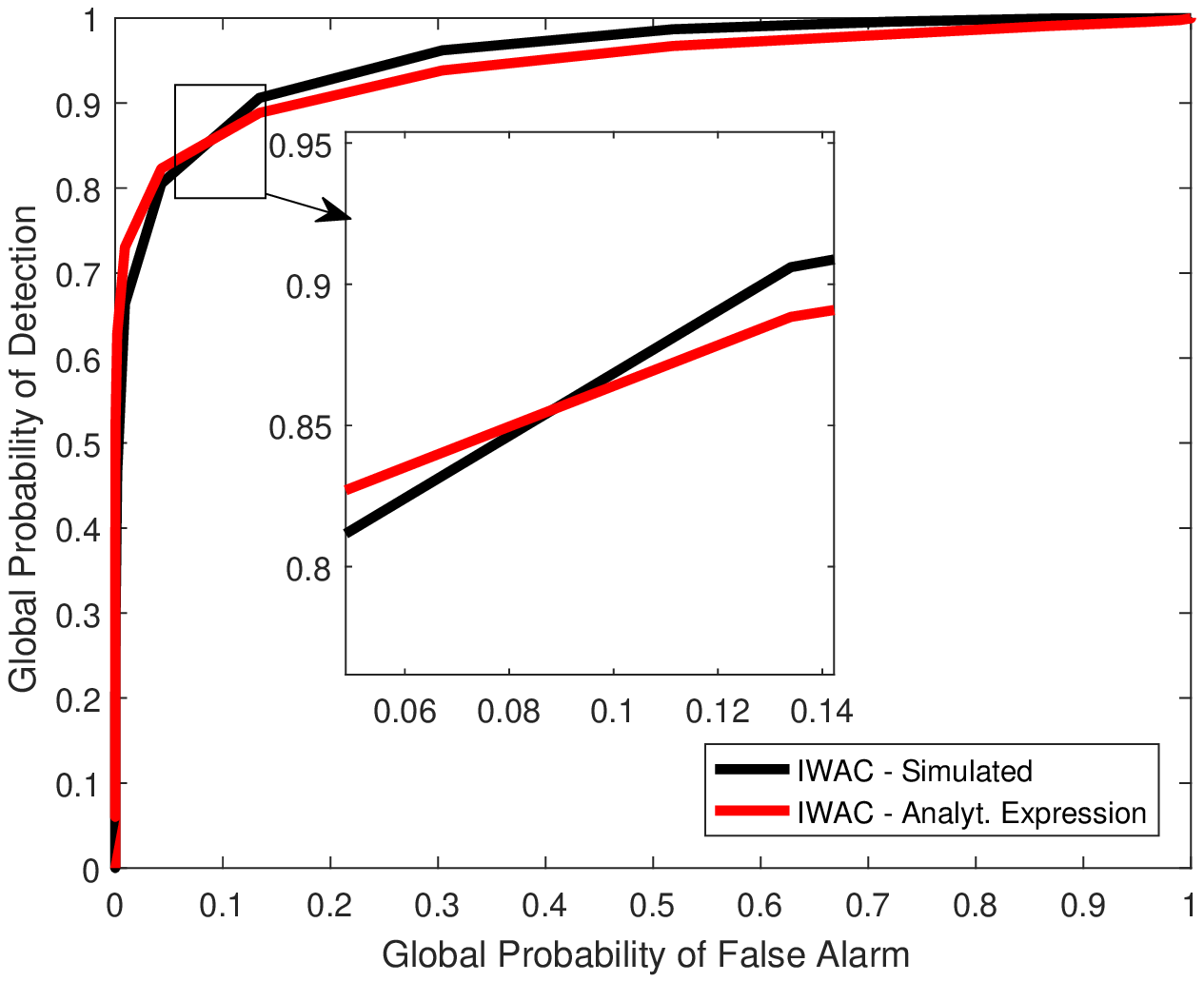}\\
	a) {\bf Fixed} Network, $6$ SUs \hspace{5cm}  b) {\bf Fixed} Network, $6$ SUs \\
 \includegraphics[width=.495\textwidth]{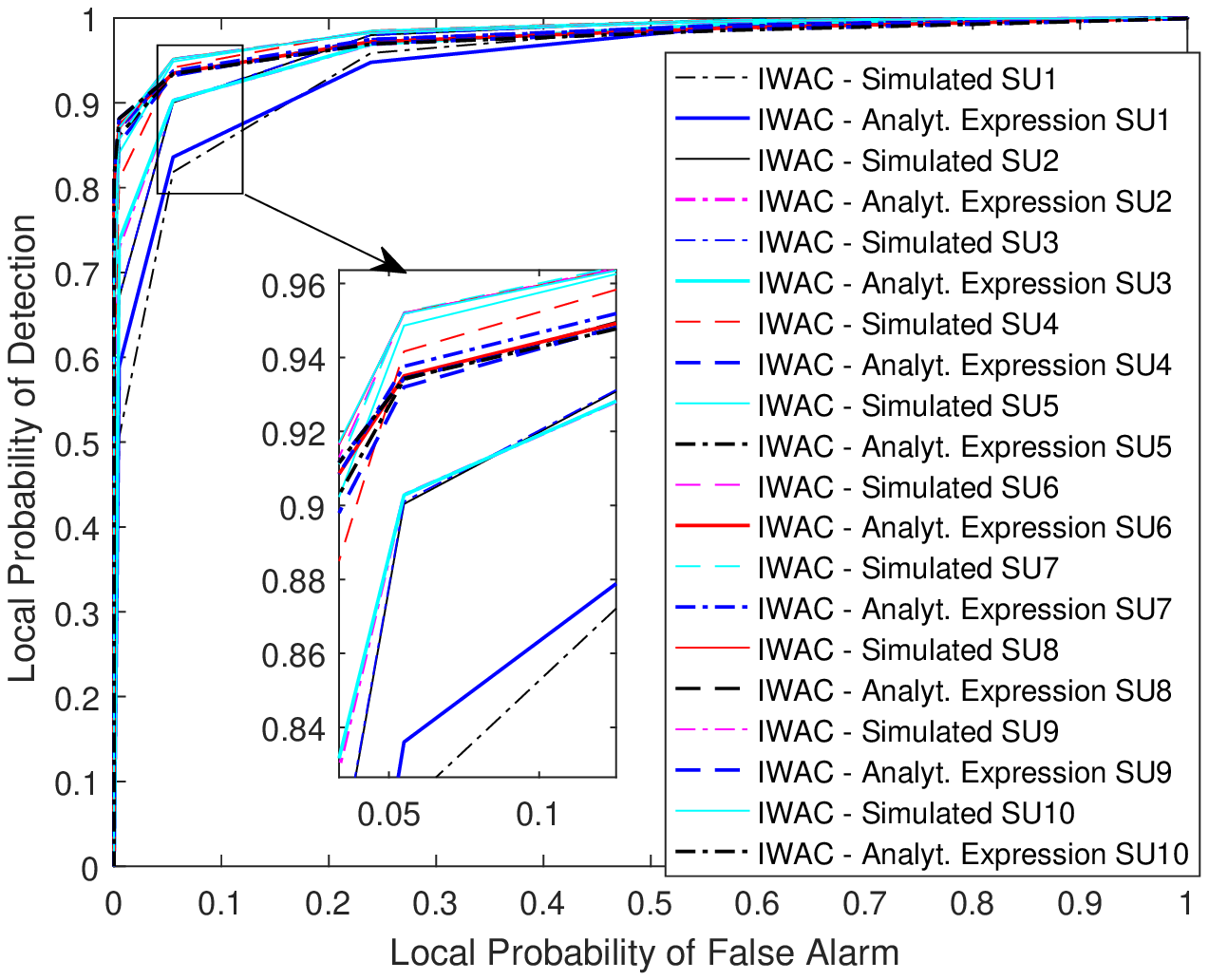}
 \includegraphics[width=.495\textwidth]{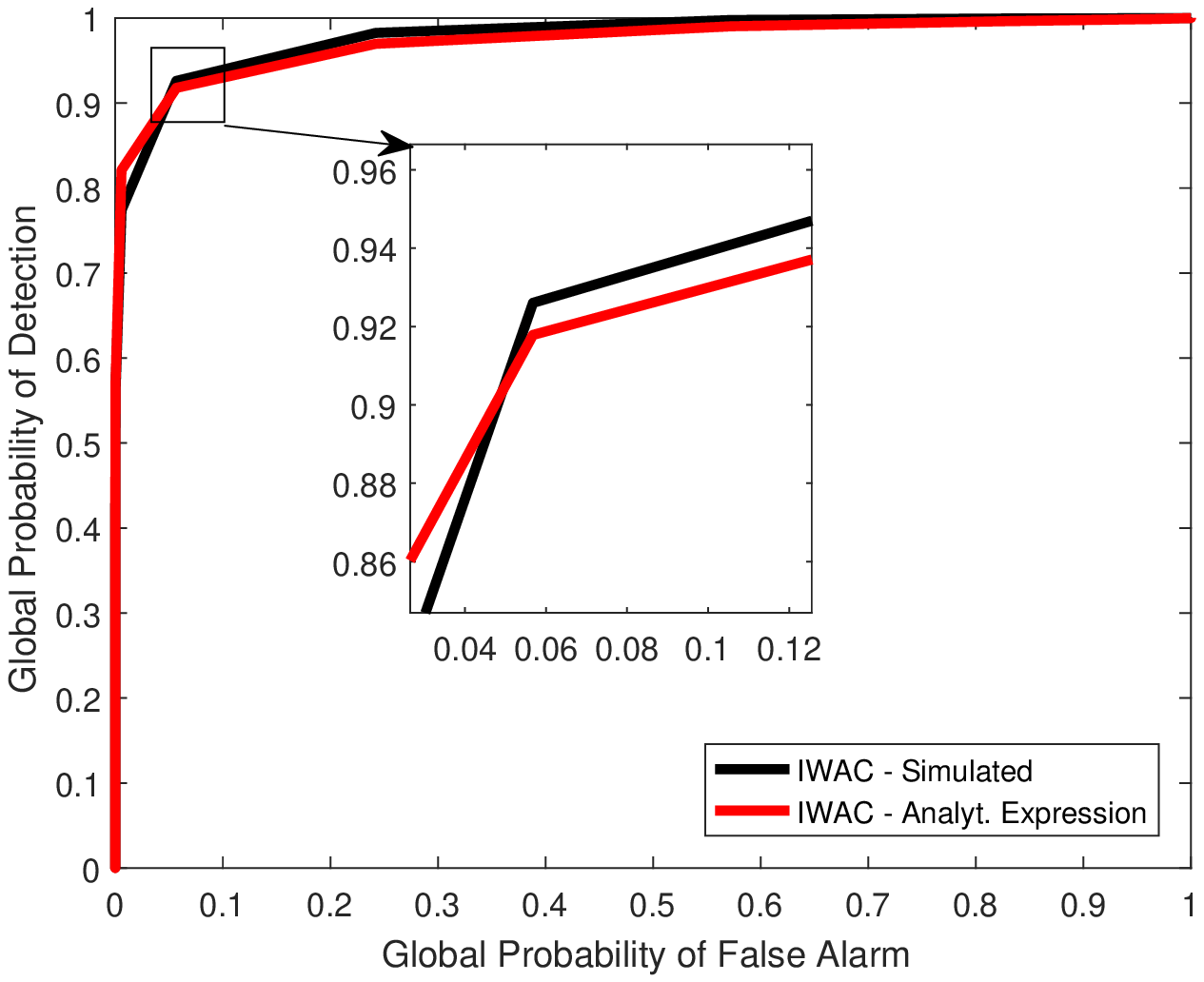}
 	c) {\bf Fixed} Network, $10$ SUs \hspace{5cm}  d) {\bf Fixed} Network, $10$ SUs \\
 \caption{Local and global ROC for $6$ and $10$ SUs under AWGN channel Scenario A.}\label{fig:LocalROC}
\end{figure}

\subsection{Computational Complexity and Average Convergence Time for Distributed AC Techniques }\label{Sec:complexity}
The {\it average convergence time} for the AC methods was established in \cite{Benezit10} considering a large number of nodes $n$ (or number of SUs) in the network as:
\begin{equation}
	\nonumber
	\mathcal{T}_{\textsc{ac}}(n) = \mathcal{O}\left(\frac{{\rm log} (n)}{1 - \rho_{2}(\mathbb{E}[{\bf P}^{T} {\bf P}])}\right), \qquad (\text{large} \,\,  n)
\end{equation} 
where $\rho_{2}$ is the second largest eigenvalue module (SLEM), the associated Perron matrix is $\bf P$ and $n$ is the number of nodes in the network (number of SUs). When $\rho_{2}(\mathbb{E}[{\bf P}^{T} {\bf P}]) \rightarrow 1$ implies that the number of secondary users in the network tends to infinity, i.e, $n \rightarrow \infty$. In this way, the {\it average convergence time} allows us to verify the dependence of the number of iterations for convergence regarding the size of the network and the AC rule chosen. In other words, the higher the value of $\rho_{2}(\mathbb{E}({\bf P}^{T} {\bf P}))$ more time is required to the consensus rule achieves convergence.

The AC complexity analysis based on SLEM values associated to the Perron matrices for each average consensus rule analysed in this work 
confirms the tendency found in our numerical results of section \ref{sec:converg}, corroborating our finding that the AC rule achieves reduced convergence time among the analysed rules, followed by our proposed WAC-AE  and IWAC rules, and finally by the WAC rule. In fact, in our paper we consider a low number of nodes in the network. Hence,  a more appropriate expression correlating the SLEM ($\rho_{2}$) and average convergence time os \cite{Boyd06}:

\begin{equation}
	\nonumber
	\widetilde{\mathcal{T}}_{\textsc{ac}} = \frac{1}{{\rm ln}\left(\frac{1}{\rho_{2}(\mathbb{E}[{\bf P}])}\right)} \qquad (\text{small or medium} \,\,  n)
\end{equation}
where $\rm ln (\cdot)$ is the natural logarithm.

The asymptotic expressions for the computational complexity of the analysed AC rules have been determined from the AC pseudo-codes (section \ref{CDSS}) and  depicted in Table \ref{tab:ACC}. As expected, the AC has the lower computational complexity among all AC distributed SS methods. The methods WAC, WAC-AE and IWAC distributed consensus methods present the same computational complexity order, resulting in a quadratic dependence with the number of SUs $N$ and a linear dependence with the number of iterations $K$.

\begin{table}[!htbp]
  \centering
  \caption{Computational Complexity for Distributed AC Algorithms}
  \label{tab:ACC}
  \begin{tabular}{cccc}
    \toprule
	 \bf AC rule &  \bf Consensus & \bf Asymptotic \\
 	\bf Algorithm  &  \bf  Method & \bf Complexity\\
    \midrule  
  	\ref{algo:AC} &  {\bf AC} & $\mathcal{O}(KN)$ \\
  	\ref{algo:WAC} &  {\bf WAC} & $\mathcal{O}(KN^2)$ \\
 	\ref{algo:WAC-AE}&   {\bf WAC-AE} & $\mathcal{O}(KN^2)$ \\
  	\ref{algo:IWAC}&  {\bf IWAC} & $\mathcal{O}(KN^2)$ \\
    \bottomrule
  \end{tabular}
\end{table}

\section{Conclusions}\label{sec:Concl}
In this paper we have proposed and analysed two new decentralised average consensus-based spectrum sensing scheme, namely IWAC and WAC-AE, and compare their performance {and complexity} with two other conventional CSS decentralised consensus-based methods (AC and WAC), as well as with other traditional centralised CSS under hard and soft combining rules. The performance comparison is made regarding the  receiver operator characteristics (ROC) and numerical {\it versus} analytical convergence. The proposed IWAC method results in similar convergence rate to the WAC-AE method.

{Regarding} ROC analysis, the WAC and WAC-AE methods demonstrate similar performance, which is {comparable} to the centralised MRC rule. Moreover, the AC method and EGC has also similar performance, which results worse than the MRC performance. Indeed, the proposed decentralised IWAC method has demonstrated ROC performance in between the centralised MRC and EGC rules.

The weighted decentralised CSS methods discussed herein result in a similar computational complexity cost, being asymptotically equal to the product of the squared number of cooperative SUs and the number of iterations, $N^2K$. Another way to evaluate the complexity {of} the AC rules is the  {\it average convergence time} based on the second largest eigenvalue module (SLEM) which is dependent on the associated Perron matrix $\bf P$ and the number of SUs $n$. The AC complexity analysis based on SLEM has confirmed the tendency found in our numerical simulation results, corroborating our conclusion that the AC rule achieves reduced convergence time among the analysed rules, followed by our proposed WAC-AE  and IWAC rules, and finally by the WAC rule. {In summary, the IWAC method results in a similar convergence rate than the WAC-AE method but slightly higher than the AC and WAC methods.}

\section*{Acknowledgement}
This work was supported in part by the National Council for Scientific and Technological Development (CNPq) of Brazil under Grants 304066/2015-0 and  308348/2016-8, and in part by CAPES -- Coordenação de Aperfeiçoamento de Pessoal de Nível Superior, Brazil (scholarship),  and by the Londrina State University - Paraná State Government (UEL).


\end{document}